\documentclass[oupdraft]{bio}
\makeatletter
\def\ps@plain{\let\@oddhead\@empty\let\@evenhead\@empty\let\@oddfoot\@empty\let\@evenfoot\@empty\let\@mkboth\markboth}
\makeatother
 
\setcitestyle{round}
\usepackage[utf8]{inputenc}
\usepackage[T1]{fontenc}
\usepackage{hyperref}
\usepackage{url}
\usepackage{booktabs}
\usepackage{amsfonts}
\usepackage{nicefrac}
\usepackage{microtype}
\usepackage{amsmath}
\usepackage{amssymb}
\usepackage{amsthm}
\usepackage{graphicx}
\usepackage{bbm}
\usepackage{array}
\usepackage{soul}
\usepackage{multirow}
\usepackage{adjustbox}
\usepackage{setspace}
\usepackage{wrapfig}
\usepackage{framed}
\usepackage{tcolorbox}
\usepackage{twemojis}
\usepackage{algorithm}
\usepackage{algorithmicx}
\usepackage{algpseudocode}
\usepackage{xurl}
\usepackage{tikz}
\usetikzlibrary{arrows, shapes.arrows, shapes.geometric, shapes.swigs, shapes.multipart, decorations.pathreplacing, decorations.pathmorphing, positioning, shapes,arrows.meta, calc, shadows}
\tikzset{rv/.style={circle,inner sep=1pt,draw,font=\sffamily},
lv/.style={circle,inner sep=1pt,fill=gray!50,draw,font=\sffamily},
fv/.style={rectangle,inner sep=1.5pt,fill=gray!20,draw,font=\sffamily},
node distance=15mm, >=stealth}

\usepackage[capitalise]{cleveref}
\usepackage{nameref}

\newcounter{mylabelcounter}

\makeatletter
\newcommand{\labelText}[2]{%
#1\refstepcounter{mylabelcounter}%
\immediate\write\@auxout{%
  \string\newlabel{#2}{{1}{\thepage}{{\unexpanded{#1}}}{mylabelcounter.\number\value{mylabelcounter}}{}}%
}%
}

\newcommand{\eggboxlabel}{\phantomsection\def\@currentlabel{\protect\twemoji{egg}}\label{box:egg}}

\newtheorem{assumption}{Assumption}

\crefname{assumption}{Assumption}{Assumptions}
\Crefname{assumption}{Assumption}{Assumptions}

\definecolor{TypeINode}{HTML}{E8590C}
\definecolor{TypeIINode}{HTML}{868E96}
\definecolor{TypeIIINode}{HTML}{1971C2}

\newcommand{\E}{\mathbb{E}}

\newcommand{\Pa}[1]{\mathrm{Pa}(#1)}
\newcommand{\Ch}[1]{\mathrm{Ch}(#1)}

\newcommand{\De}[1]{\mathrm{De}(#1)}
\newcommand{\Nd}[1]{\mathrm{Nd}(#1)}

\newcommand{\indep}{\mathrel{\text{\scalebox{1.07}{$\perp\mkern-10mu\perp$}}}}
\newcommand{\sD}{\mathrm{D}}

\newcommand{\sC}{\mathrm{C}}

\begin{document}

\title{
Expert-Guided g-computation with Large Language Models for Estimating Causal Effects on Timings:\\
Applications to Hospital Quality Improvement
}

\author{Patrick Vossler$^1$, Jialin Ouyang$^1$, F.~Richard Guo$^2$, Anran Huang$^1$,\\\vspace{-0.3em}
Ali Shojaie$^3$, Lucas Zier$^1$, Fan Xia$^{1*}$, Jean Feng$^{1*\dagger}$}

\address{
$^1$University of California, San Francisco,
$^2$University of Michigan, Ann Arbor\\
$^3$University of Washington, Seattle\\
$^*$ Co-senior authors\\
$^\dagger$ Corresponding author: \texttt{jean.feng@ucsf.edu}, 550 16th Street, San Francisco, CA, 94158
}

\markboth%
{Vossler et al.}
{Expert-guided g-computation with LLMs}

\maketitle

\begin{abstract}
{
Hospital quality improvement (QI) programs are routinely faced with multiple candidate interventions to optimize hospital flow, but existing methods are limited in their ability to estimate and rank the causal effects of such interventions.
This work focuses on one of the most standard hospital metrics---the average length of stay (LOS)---and its associated causal estimand---the average time saved.
To characterize this causal effect, qualitative approaches rely on expert judgment to map complex patient trajectories, making them susceptible to cognitive biases; conversely, quantitative approaches rely on data-driven models, which fail when interventions are entirely hypothetical with no historical data and/or complex causal mechanisms that cannot be learned through data and instead require clinical reasoning.
Here we propose \textit{expert-guided g-computation} or egg-computation, which builds on the complementary strengths of qualitative and quantitative approaches and connects Gantt charts commonly used to map patient trajectories with the causal DAG literature.
We first introduce a causal model over Gantt charts and establish causal identification using a variant of g-computation that only seeks expert input for components unidentifiable from data.
To make egg-computation practical, we develop a Large Language Model (LLM)-assisted pipeline that reliably scales up expert reasoning.
Simulations show that egg-computation improves accuracy compared to conventional causal inference methods in complex settings where patients have diverse causal structures and intervention mechanisms.
In a study of eleven candidate QI interventions at an urban safety-net hospital, the LLM pipeline generated graphs and time-saving estimates that were highly concordant with those of human experts.
Beyond healthcare, egg-computation is a broadly applicable framework for estimating average time saved for candidate interventions whose causal mechanisms can be represented using Gantt charts.
}

{{Causal inference, g-computation, Large language models, Gantt charts}}
\end{abstract}

\section{Introduction}\label{sec:intro}
Hospital Quality Improvement (QI) programs seek to improve the healthcare delivery process and key hospital performance metrics like length of stay (LOS).
Numerous QI interventions have been proposed in the literature, including better triaging of imaging studies and medical procedures \citep{Xu2026-new, Hurlen2010-new}, balancing workload across a week \citep{Vidal-Carreras2022-bg}, and using Artificial Intelligence (AI)-based clinical decision support tools \citep{Levin2021-new, Wornow2023-uk}.
Recently, Large Language Models (LLMs) have even been used to generate a bank of candidate QI initiatives by identifying delays/bottlenecks most commonly mentioned in hospital records \citep{Vossler2026-ew}.
Yet QI programs can only implement a handful of interventions---often just one or two---and must therefore prioritize by expected impact, which ultimately requires answering a \textit{causal} question.
In this work, our focus is on characterizing the causal effect on LOS, i.e., \textit{average time saved}, as LOS is among the most widely-used metrics of hospital efficiency, is linked to patient outcomes, and factors into reimbursement \citep{Rojas-Garcia2018-dn, Ghosh2023-fv, Liang2026-uq, The-Press-Association2014-zs}.

To anticipate how a QI intervention would affect LOS, QI teams commonly rely on semi-structured qualitative analyses such as staff interviews, manual chart review, and Lean healthcare methods \citep{Catalyst2018-ho, D-Andreamatteo2015-zg}. 
A central tool in these qualitative analyses is the classical Gantt chart (or, equivalently, the value stream map) \citep{Rother1999-new, Vidal-Carreras2022-bg}, where a patient's hospital journey is broken down into tasks and their dependencies.
For example, consider the following early discharge-planning QI intervention based on findings in \citet{Vossler2026-ew}:
\begin{quote}
    \textit{Any patient with complex discharge needs, as identified on the first day of their hospital stay, will be evaluated for discharge planning and meet with social work or case management within their first 24-48 hours. If the patient agrees, appropriate discharge resources will be contacted within 24 hours after meeting with the patient.
    Exclusion criteria include patients who cannot meet because of medical or legal reasons.}
\end{quote}
\noindent 
By creating a Gantt chart of the patient's journey, hospital QI teams can use their domain expertise to reason about how a described intervention would likely shift event timings and the overall LOS.
While experts may be accurate at estimating how an intervention shifts the events it directly targets (e.g., the time when the patient meets with social work or case management), its likely impact on more downstream events is harder to predict, as those event timings depend on interactions among multiple preceding events and cannot be determined through simple fixed rules.
For example, an earlier social work consultation may not substantially accelerate discharge if the information needed to determine the appropriate disposition is not yet available.
Similarly, moving a consultation from Monday to the preceding Friday may leave LOS unchanged if the anticipated discharge institution (e.g., a skilled nursing facility) does not process admissions over the weekend.

Given the limitations of qualitative analyses, QI teams often conduct complementary quantitative analyses.
The most common approach is to fit a regression model; however, such models generally estimate association not causation, and traditionally rely solely on tabular features, which alone may be insufficient for eliminating confounding bias \citep{Steverson2023-hu, Zeng2022-new}.
While a formal causal framework is desired, developing such a framework is non-trivial:
(i) Understanding a QI intervention and its effects requires complex clinical reasoning, and much of the relevant information resides in unstructured clinical notes rather than the tabular data.
(ii) The commonly assumed positivity assumption is difficult to satisfy as proposed QI interventions are, by definition, new policies that the hospital has never implemented and no patients have been exposed to \citep{Westreich2010-vc, Petersen2012-new}.
(iii) Because QI interventions often cover a broad and heterogeneous patient population, constructing a unified causal diagram that encompasses the all patient trajectories is a near-impossible task, as even the set of hospital events documented in clinical notes is open-ended. Furthermore, even if such a causal diagram existed, the available data is limited for accurately estimating the conditional relationships in the graph.

\begin{figure}
    \centering
    \includegraphics[width=\linewidth]{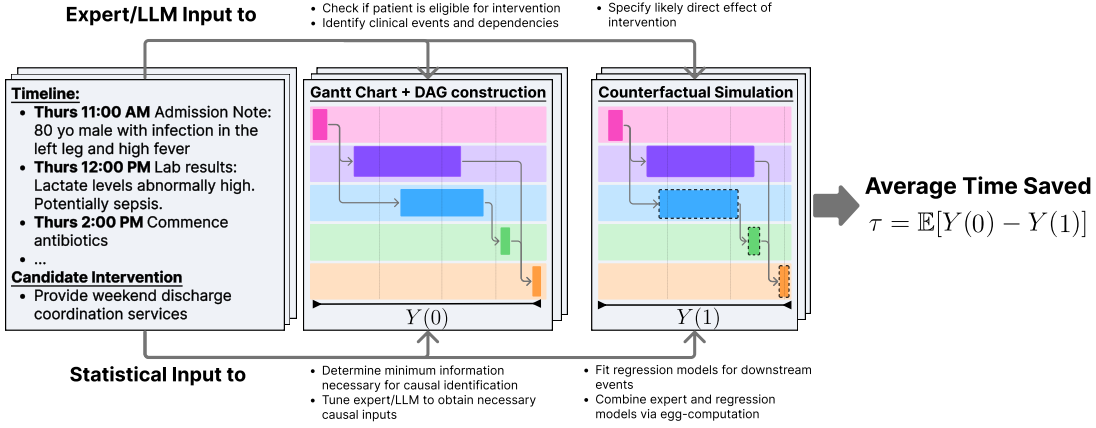}
    \vspace{-1.1cm}
    \caption{
    \textbf{Overview of expert-guided g-computation.}
    To estimate the effect of a candidate Quality Improvement (QI) intervention on the average length of stay (LOS) among eligible hospitalizations, the expert, or an aligned LLM, reads unstructured texts describing the candidate intervention and patient records (left) and uses expert reasoning to extract patient-specific Gantt charts, which we prove can be endowed with probabilities and causal semantics (middle).
    Through a combination of expert reasoning and statistical modeling, we simulate the counterfactual charts under the candidate intervention (right) and estimate the target estimand ``Average Time Saved'' $\tau = \mathbb{E}[Y(0) - Y(1)]$ where $Y(1)$ and $Y(0)$ are the LOS with or without intervention, respectively.
    Dashed outlines mark events whose timings are shifted in the counterfactual trajectory.
    }
    \label{fig:egg_overview}
\end{figure}

To resolve these technical gaps, this work combines qualitative and quantitative approaches to conduct what we refer to as ``expert-guided causal inference:'' expert judgment is elicited solely in areas where experts are accurate but data-driven modeling is not, while statistical models are employed where data-driven models are accurate but experts may not be (Figure~\ref{fig:egg_overview}).
The final causal contrast reflects the likely effect of an intervention based on both the data and expert opinion and, for settings where the expert is accurate, we establish identifiability of the true causal effect.
Proving this result requires carefully translating Gantt charts into directed acyclic graphs (DAGs) endowed with formal probabilistic and causal semantics for event processes (Section~\ref{sec:gantt}).
Given the resulting DAGs, we then introduce a variant of g-computation \citep{Robins1986-kr} that we refer to as \textit{expert-guided g-computation} or egg-computation, where expert manipulation of the Gantt chart is combined with statistical sampling of downstream events to generate counterfactual patient trajectories (Section~\ref{sec:egg}).
Because a fully human-powered approach is likely prohibitively time- and resource-intensive, we developed an LLM-assisted pipeline that supplies the expert reasoning needed to conduct egg-computation at scale (Section~\ref{sec:llm_scaling}).
This LLM pipeline can be easily adapted to new clinical settings or interventions, as it only depends on prompting rather than model retraining or fine-tuning.

To demonstrate the efficacy of the proposed egg-computation framework, we conduct extensive empirical evaluation.
We begin with comparing egg-computation with conventional causal inference methods in simulation studies, which show that egg-computation is more accurate when the causal structure or intervention mechanism varies substantially across patients.
Then in a real-world study of eleven candidate QI interventions at an urban safety-net hospital, we evaluate the use of LLMs to scale human-expert-guided causal inference.
We find that LLM extractions are highly concordant with human experts and even as accurate as outcome regression models, leading to highly concordant time saving estimates.
Clinical experts also found that LLM-generated reasoning traces provide useful insights into how well an intervention works and how it might be better designed.
Finally, while our focus is on estimating the effects of hospital QI interventions on patient trajectories, the methodology is broadly applicable to settings where Gantt charts are suitable for describing task dependencies and process timings \citep{Wilson2003-zq, Clark1922-ih} and the Related Work section in the Appendix places this contribution in the broader literature.

\section{Causal models for Gantt charts} \label{sec:gantt}

A Gantt chart (see Figure~\ref{fig:egg_overview}) is a collection of horizontal bars corresponding to events, where the horizontal axis corresponds to time.
Directed edges between bars represent dependencies between events, which are often categorized into start-start, finish-start, start-finish, and finish-finish relationships \citep{project2000guide}.
If we ignore the time axis and divide an event into a pair of (\texttt{start}, \texttt{end}) events with the \texttt{end} depending on the \texttt{start}, a Gantt chart can be drawn as a Directed Acyclic Graph (DAG), where each node represents the time of an event and directed edges describe the dependency between event times.
To represent the start and completion of a multi-step process, a Gantt chart must include source and sink nodes, which in our setting correspond to hospital admission and discharge/death, respectively.

While a Gantt chart corresponds to a DAG---a representation commonly used for causal models of various kinds \citep{wright1934method,pearl1995causal,Robins1986-kr}---this correspondence alone does not endow it with the causal semantics needed to address the inferential question of interest.
Rather, formalizing the causal model associated with a Gantt chart requires clarifying what edges mean exactly in this context, 
how an intervention changes the event process, and which features of that process are assumed to remain unchanged.
To this end, we introduce a causal model in terms of event waiting times that enables identification of intervention effects within subpopulations defined by $X$.

\paragraph{Notation.} We use upper-case letters (e.g., $X,Y,Z$) for random variables or vectors and use lower-case (e.g., $x,y,z$) for values.
We use Roman capital letters (e.g., $\sC, \sD$) for sets, which can appear as subscripts for indexing random vectors. 
For a node $k$ in a DAG, we use the standard graphical notation: $\Pa{k}:=\{j: j \to k\}$ for parents, $\Ch{k}:=\{j: j \leftarrow k\}$ for children, $\De{k}:=\{k\} \cup \{j: k \to \dots \to j\}$ for descendants, and $\Nd{k}:=\text{De}^{c}(k)$ for non-descendants. 
We use `$=_{d}$' to denote equality in distribution. 

\subsection{Causal model over waiting times} \label{sec:causal-wait}

\begin{figure}[h]
\centering
\includegraphics[width=\linewidth]{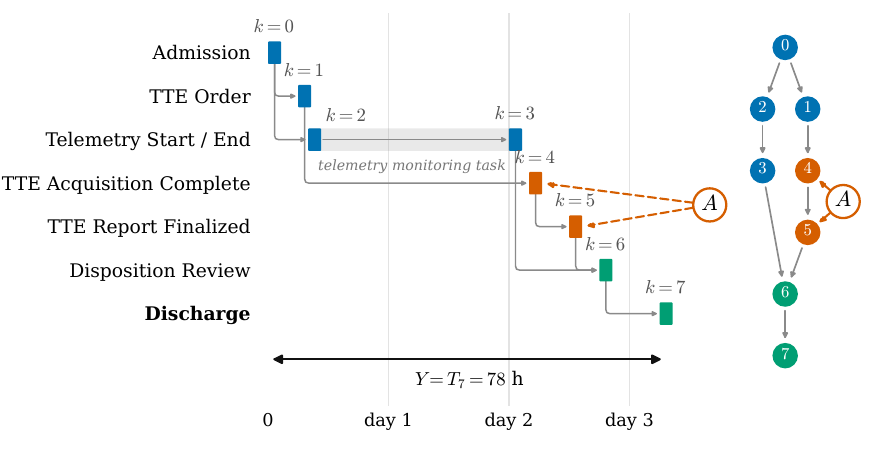}
\definecolor{typeI}{HTML}{D55E00}
\definecolor{typeII}{HTML}{0072B2}
\definecolor{typeIII}{HTML}{009E73}
\newcommand{\typeswatch}[1]{\protect\raisebox{-0.2ex}{\protect\tikz{\protect\node[rectangle, rounded corners=1pt, fill=#1, minimum height=2.2mm, minimum width=5mm, inner sep=0pt]{};}}}
\vspace{-1cm}
\caption{\textbf{An example inpatient trajectory drawn as a Gantt chart with its causal DAG.}
Each row is one event, placed at its event time in hours from admission, and gray arrows are the DAG edges; events are labeled $k = 0, \dots, K$ in topological order, with $K = 7$.
Node $A$ is not an event: it denotes a binary intervention that directly affects the waiting times of the transthoracic echocardiogram (TTE) acquisition and report (dashed arrows).
Bar colors give each event's relation to $A$: $k \in \{0,1,2,3\}$ are non-descendants of $A$ (\typeswatch{typeII}), $k \in \{4,5\}$ are children of $A$ (\typeswatch{typeI}), and $k \in \{6,7\}$ are descendants of $A$ that are not children of $A$ (\typeswatch{typeIII}).
Event `Admission' defines the start of wall time $T_0 = W_0 = 0$; `Discharge' defines the outcome of interest $Y := T_7$.
The observed data is generated under $A=0$. 
}
\label{fig:ex-gantt}
\end{figure}

Consider a random patient with baseline covariates $X \in \mathcal{X}$. 
The patient's Gantt chart, when ignoring the specific event times, can be described as a DAG over a finite set of events, which we denote as $\sD = \{0,1,\dots,K\}$; see \cref{fig:ex-gantt} for an example.
Given the DAG, an event $k$ with $\Pa{k} \neq \emptyset$ can only occur \emph{after} all the events in $\Pa{k}$ have occurred. 
We assume the labels are topologically ordered so that $j \in \Pa{k}$ implies $j < k$.
For simplicity, we will suppose both the baseline covariates $X$ and the set of events $\sD$ are \emph{conditioned on},  thereby fixing the DAG associated with the Gantt chart. 
In other words, the random patient is considered drawn from the subpopulation with the given $X$ and $\sD$.

For each event $k$, let $T_k$ be the \emph{wall time} at which $k$ occurs.
Let $W_k$ be the \emph{waiting time} for its occurrence, namely the time taken for it to occur since the occurrence of all of its preceding events. 
The wall time and waiting time are related through 
\begin{equation} \label{eq:T-W}
 T_k = \begin{cases} \max_{j \in \Pa{k}} T_j + W_k, \quad & \Pa{k} \neq \emptyset \\ 
W_k, & \Pa{k} = \emptyset \end{cases}.
\end{equation}
Given $W_k$ for each $k$, one can use \eqref{eq:T-W} to recursively compute $T_0, T_1, \dots, T_K$; conversely, given $T_k$ for each $k$, \eqref{eq:T-W} directly recovers $W_0, W_1, \dots, W_K$.
We let $T_k$ and $W_k$ take values in $[0, +\infty]$, where $T_k = +\infty$ means the event $k$ never occurs. 
If $T_k = +\infty$ for some event $k$, the ensuing events also never occur, i.e., $T_j = +\infty$ for all $j \in \De{k}$. 
We designate two special events: $k=0$ as the source event (`Admission') with wall time $T_0 = W_0 := 0$ and $k=K$ as the sink event (`Discharge'), which defines the outcome of interest $Y := T_K$.
All events are required to be an ancestor of event $K$, i.e., $k \to \dots \to K$ for each $k \neq K$; otherwise we can discard such events without loss of generality.

In our event-time setting, the DAG model posits the \emph{local Markov property} in terms of conditional independence
\begin{equation} \label{eq:local-markov}
W_k \indep T_{\Nd{k}} \mid T_{\Pa{k}}, X, \sD \quad k=0,1,\dots,K.
\end{equation}
This property should be differentiated from that of a standard Bayesian network model.
By \eqref{eq:T-W}, wall times obey the ordering $T_k \leq T_j$ for all $j \in \De{k}$; therefore, \eqref{eq:local-markov} cannot be restated as $T_k \indep T_{\Nd{k}} \mid T_{\Pa{k}}, X, \sD$ nor $W_k \indep W_{\Nd{k}} \mid W_{\Pa{k}}, X, \sD$ in general.

A QI policy is defined using a free-text description along with inclusion and exclusion criteria, like the early discharge intervention example described in \cref{sec:intro}.
To represent the causal impact of such a policy, the Gantt DAG is endowed with an additional non-event node $A$ that represents the \emph{binary intervention}:
$A=0$ represents the current policy under which our data is observed, while $A=1$ represents a new QI policy, which we do not observe but want to infer. 
We refer to the direct children of an intervention, denoted $\Ch{A}$, as \textit{intervention targets}, whose waiting times are directly modified by the new policy.
We assume $A$ is \emph{exogenous}, represented by $\Pa{A} = \emptyset$ in the augmented Gantt chart.
Every event $k$ is associated with the potential waiting times $W_k(0), W_k(1) \in [0, +\infty]$ which correspond to potential wall times $T_k(0), T_k(1) \in [0, +\infty]$ through relationships defined in a similar manner to \eqref{eq:T-W}, i.e.,
\begin{equation} \label{eq:T-W-a}
 T_k(a) = \begin{cases} \max_{j \in \Pa{k}} T_j(a) + W_k(a), \quad & \Pa{k} \neq \emptyset \\ 
W_k(a), & \Pa{k} = \emptyset \end{cases}, \qquad a=0,1,
\end{equation}
where again we let $W_0(0) = W_0(1) = 0$.
Implicit in our notation is the simplifying assumption of no interference or SUTVA \citep{rubin1980randomization}.
For example, if $A=1$ changes the CT schedule, we assume each patient is affected separately, even though in reality they may not be, due to resource constraints; see also \cref{sec:discuss} for more discussion.

$Y(a) := T_K(a)$ for $a=0,1$ is our potential outcome of interest measuring the total wall time from source to sink. 
Thus, the individual treatment effect is $Y(0) - Y(1)$, namely the time \emph{saved} by the new QI policy. 
As we will see in \cref{sec:egg}, the individual treatment effect is not identified.
Instead, we are interested in inferring the population-level average treatment effect
\begin{equation}\label{eq:ATE}
\tau := \E [Y(0) - Y(1)] = \E\left\{ \E[Y(0) - Y(1) \mid X, \sD] \right\},
\end{equation}
where the outer expectation is averaging over both $X$ and $\sD$.

Following the single-world/FFRCISTG causal semantics \citep{Robins1986-kr,richardson2013single} associated with a DAG, we associate our Gantt chart with a causal model defined by \eqref{eq:T-W-a} along with the following assumptions.
Recall that $T_k$ and $W_k$ are observed only under the current policy $A=0$. 
This differs from common clinical trial and observational study settings, where one observes both $A=0$ and $A=1$ arms. 

\begin{assumption}[Consistency] \label{assump:consistency}
For $k=0,1,\dots,K$, we have $T_k = T_k(0)$ and $W_k = W_k(0)$.
\end{assumption}
\begin{assumption}[Exogenous intervention] \label{assump:exo}
For $a=0,1$, we have $W_0(a), W_1(a), \dots, W_K(a) \indep A \mid X , \sD$.
\end{assumption}
\begin{assumption}[Exclusion restriction] \label{assump:exclusion}
For $k=0,1,\dots,K$,
\begin{enumerate}
\itemsep-1em 
\item if $k \in \Nd{A}$, then $W_k(1) = W_k(0)$.
\item if $k \in \De{A} \setminus \Ch{A}$, then  $W_k(0) \mid T_{\Pa{k}}(0), X, \sD =_{d} W_k(1) \mid T_{\Pa{k}}(1), X, \sD$.
\end{enumerate}
\end{assumption}

\begin{assumption}[Local Markov property] \label{assump:local-markov}
For $k=0,1,\dots,K$, we have $W_k(a) \indep T_{\Nd{k}}(a) \mid T_{\Pa{k}}(a), X, \sD$ for $a=0,1$, where we let $A$ be excluded from $\Nd{k}$ and $\Pa{k}$.
\end{assumption}

Under \cref{assump:consistency}, we have $Y(0) = T_K$ and hence $\tau = \E T_K -  \E\left\{ \E[T_K(1) \mid X, \sD] \right\}$ by \eqref{eq:ATE}, where $T_K$ is observed. 
Therefore, we can estimate $\tau$ if $T_K(1)$ can be imputed based on  observed data under $A=0$.
By the assumptions above, the distribution of the counterfactual timings can be factorized as
\begin{align}
    p(T_\sD(1) \mid X, \sD) = \prod_{k = 0}^{K} p(T_k(1) \mid T_{\Pa{k}}(1), X, \sD)
    \label{eq:counter-factor}
\end{align}
so \cref{sec:egg} will discuss how to obtain each factor in \eqref{eq:counter-factor} so to recursively impute $T_\sD(1)$. 

\begin{remark}
For ease of exposition, our presentation has treated the event set $\sD$ in the Gantt chart as fixed and our notation has implicitly assumed that $\sD(0) = \sD(1)$ for each patient.
While this can be overly restrictive, the following causal identification results still hold if we relax this assumption to distributional equivalence, i.e., $\sD(0) \mid X =_d \sD(1) \mid X$, with \cref{assump:consistency,assump:exo,assump:exclusion,assump:local-markov} adjusted accordingly.
Such an assumption is plausible primarily for \textit{operational} interventions in healthcare, such as speeding up the timing of an imaging or consult order or increasing the availability of ancillary services over the weekend.
In contrast, interventions that drastically change a patient's trajectory (e.g., swapping the treatment from chemotherapy to surgery) would not satisfy this assumption and should not be analyzed through this framework.
\end{remark}

\section{{\scalebox{1.8}{\twemoji{egg}}} \underline{E}xpert-\underline{G}uided \underline{g}-computation, plus LLMs}
\label{sec:egg}

Having established patient-specific Markovian DAGs as valid causal models for Gantt charts, we now further leverage this connection between Gantt charts and causal DAGs.
In particular, we note that expert-manipulation of timestamps in a Gantt chart shares many similarities to using g-computation to simulate counterfactual node values on causal DAGs \citep{Robins1986-kr}.
Merging these ideas, we introduce ``expert-guided g-computation'' (\textit{egg}-computation): we can identify the causal estimand $\tau$ defined in \eqref{eq:ATE} through a combination of data-driven modeling and expert input, where the latter supplies information about the interventional distribution that cannot be learned from observational data.

\subsection{Assumptions and Identifiability}

Because $\E Y(0)$ can be identified by taking a simple average of the observed outcomes $Y$, this section focuses on identifying $\E Y(1)$.
Our general strategy is to simulate the counterfactual outcome $Y(1)$ by recursively imputing the intermediate waiting times $W_k(1)$ for $k\in \sD$ and iteratively applying definition \eqref{eq:T-W-a}.
This begins with the expert supplying accurate patient-specific DAGs.

\begin{assumption}[Expert DAG specification] \label{ass:dag_recovery} 
Given the patient's baseline covariates $X$ and the observed event sequence, the Gantt DAG posed by the expert is correctly specified. 
That is, $\Pa{k}$ for each event $k$ and the intervention targets $\Ch{A}$ are correctly specified. 
\end{assumption}

Although this assumption is untestable, we believe it is plausible for many operational interventions, because clinical notes combined with tabular data often contain sufficiently rich detail for experts to understand the causal relationships between different events \citep{Vossler2026-ew}.$^*$
\footnotetext{Indeed, this is why clinician chart review is a common standard for making causal judgments, like determining the cause of an adverse event \citep{Griffin2009-new, Classen2011-new} or the reason for a readmission \citep{Auerbach2016-new}.}
For instance, tabular data may provide a sequence of events and timestamps (e.g., ``10:30 AM MRI completed, 11:45 AM IV antibiotics started'') while clinical notes often describe how these events are linked (e.g.,  ``starting IV antibiotics for the deep-tissue infection seen on MRI'').
In contrast, this assumption may be less plausible if the intervention changes complex treatment regimens.
Thus, prior to conducting egg-computation, one should assess whether Assumption~\ref{ass:dag_recovery} can be (approximately) satisfied.

Given the DAG, the next step is to determine how best to impute the waiting times.
Since we have access to both observational data and an expert, we should choose the optimal information source for  different nodes, which we divide into three mutually exclusive types:

\noindent \textbf{Type I: Events that are not descendants of intervention $A$, i.e., $k \in \Nd{A}$.}
This is the simplest option.
By \cref{assump:exclusion}, waiting times of Type I nodes can be directly copied from the factual timeline, i.e., $W_k(1) = W_k$.
As these nodes are completely unaffected by the intervention (as all ancestors are also unaffected), we also have $T_k(1) = T_k$ for all Type I nodes.

\noindent \textbf{Type II: Events that are descendants of $A$ but not direct children, i.e., $k \in \De{A} \setminus \Ch{A}$.}
Assuming the conditional distribution of the observed waiting times is correctly learned, i.e., $W_k \mid T_{\Pa{k}}, X, \sD$, the model can be used to directly impute $W_k(1)$ since these two distributions are equivalent by \cref{assump:exclusion}.
We state this modeling assumption below and discuss estimation in a later section.

\begin{assumption}[Timing Model] \label{assump:timing_model}
We have a Type~II timing model that equals the true conditional distribution $p\left(W_k \mid T_{\Pa{k}}, X, \sD\right)$ for all $k \in \De{A} \setminus \Ch{A}$.
\end{assumption}

\noindent \textbf{Type III: Events that are direct children of intervention $A$, i.e., $k \in \Ch{A}$}
The waiting times directly affected by a new QI policy are never observed, so we must rely on expert guidance.
To achieve identification, we must assume our expert is accurate at imputing, as formalized below:
\begin{assumption}[Expert Counterfactual Simulation]
\label{assump:expert}
The expert-specified distribution $p_{\text{expert}}$ equals the counterfactual waiting time distribution $p\left(W_k(1) \mid T_{\Pa{k}}(1), X, \sD \right)$ for all $k \in \Ch{A}$.
\end{assumption}
\noindent In practice, this assumption need not hold exactly. If the intervention's true direct effects are unknown, the expert can hypothesize a set of functions that vary the size of the intervention's direct effects and weight the resulting estimates against their prior belief.

Now that counterfactual timestamps can be simulated for all three types of nodes along the patient's DAG in topological order, we summarize the entire pipeline from observed Gantt charts to sampling counterfactual outcomes in Box~\ref{box:egg}:
\begin{tcolorbox}[title={Box \twemoji{egg}: Expert-guided g-computation}]
\eggboxlabel
\begin{enumerate}
\item \textbf{DAG Construction}: The expert first determines whether the intervention, based on its inclusion and exclusion criteria, applies to the hospitalization.
If not, the intervention node $A$ has no children and the counterfactual trajectory equals the observed one.
Otherwise, the expert analyzes the observed event sequence under $A = 0$ and constructs the causal DAG per Assumption~\ref{ass:dag_recovery}.
Each node is classified as Type~I, II, or III based on its relationship to the intervention.

\item \textbf{Expert-guided Simulation}: Following the topological order of the DAG, generate counterfactual event times $t'$ for each patient under $A=1$:
\begin{itemize}
\item For Type I nodes: copy observed times directly from factual sequence
\item For Type II nodes: predict timing using models trained on observational data, conditioned on the (now-modified) parent timings
\item For Type III nodes: sample from the expert-supplied distribution $p_{\text{expert}}$
\end{itemize}

\item[] \hspace{-0.3cm} \textbf{Output}: Take the average difference between the factual outcome $Y_i$ and the imputed counterfactual outcome $Y_i'$, i.e.,
\begin{align}
    \hat{\tau} = \frac{1}{n} \sum_{i=1}^{n} \left(Y_i - Y_i'\right)
    \label{eq:ate_estimator}
\end{align}
\end{enumerate}
\end{tcolorbox}

The following theorem gives the guarantee of \textit{egg}-computation; see Appendix~\ref{sec:proofs} for its proof. 

\begin{theorem}[Correctness of egg-computation]
\label{thm:correctness}
Suppose \cref{assump:consistency,assump:exo,assump:exclusion,assump:local-markov,ass:dag_recovery,assump:timing_model,assump:expert} hold. 
Then, $Y_i'$ generated by the \textit{egg}-computation is identically distributed as the counterfactual time $Y_i(1)$. 
Consequently, $\hat{\tau}$ in \eqref{eq:ate_estimator} is an unbiased estimator of the average treatment effect $\tau$ in \eqref{eq:ATE}.
\end{theorem}

\subsection{Scaling up with LLMs}
\label{sec:llm_scaling}

Compared to classical causal inference frameworks, egg-computation requires substantially more expert input.
In some sense, the ``amount'' of input needed for egg-computation scales with the number of observations: based on the intervention description and each patient's clinical chart, the expert must (i) decide whether and when the intervention applies, (ii) construct a DAG and (iii) assist with event time simulation.
While this allows egg-computation to conduct expert-guided causal reasoning in substantially more complex settings, a key question is how one can feasibly supply this expert input in practice.

With the major advancements in reasoning abilities of LLMs, it is now practically feasible to scale egg-computation.
Below, we discuss how LLMs are prompted to supply expert input for the two steps of egg-computation: eligibility determination and DAG construction in Step 1 and expert-guided simulation in Step 2.
A more in-depth discussion of the pipeline's implementation is included in Appendix~\ref{sec:dag_pipeline}, and all prompts will be available in an open-source software package.
\subsubsection{Step 1}
\label{sec:step_1}

Step 1 of egg-computation determines whether the intervention applies to a given hospitalization at all and, if so, how to construct the DAG.
Determining eligibility is a relatively simple information extraction and reasoning task \citep{Agrawal2022-me, Singhal2023-sc, Kanjee2023-new}, so an LLM needs only minimal prompt tuning to accurately assess whether a patient is eligible for an intervention.
However, DAG construction is harder.
A single zero-shot LLM prompt has only been shown to accurately extract small causal DAGs ($\le 5$ nodes) in well-understood settings \citep{Kiciman2024-gz}, but this strategy is unlikely to yield valid, let alone accurate, DAGs when both the number of DAGs and the number of nodes per DAG are much larger.

Substantial engineering is required to translate manual, and often highly iterative, approaches of DAG construction into a set of LLM prompts with reliable performance.
At a high-level, our LLM pipeline is composed of five steps (Figure~\ref{fig:dag_construction}), where each step is one LLM call that refines the set of nodes and/or edges in the working DAG.
The first step prompts the LLM to reason about the patient's clinical course and draft the causal structure; the second and third steps formalize this draft as a Gantt chart by deciding the nodes and then the edges, respectively; and the final two steps construct the DAG in a structured format.
This multi-step process was designed to efficiently spend the LLM's reasoning tokens by prioritizing different aspects of DAG construction.
Early steps in the pipeline prioritized clinical and causal reasoning and later steps prioritized structuring the output to match the expected encoding and double-checking the DAG structure.
Throughout, a key design principle when crafting the LLM prompts was to use formal statistical language, as this precision encourages the LLM to use similarly rigorous reasoning.

\begin{figure}
\centering
\includegraphics[width=0.99\linewidth]{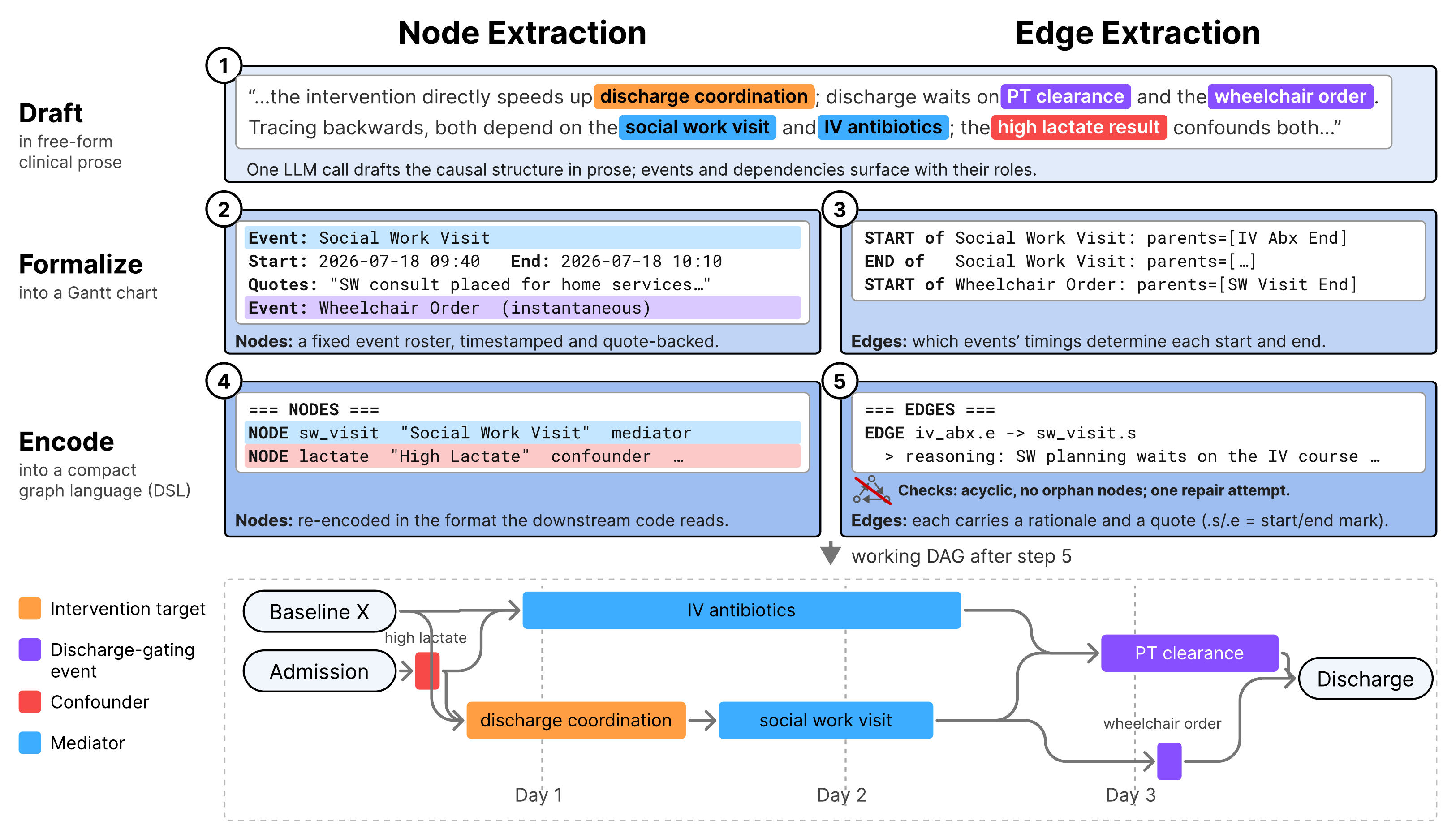}
\caption{\textbf{Five-step LLM pipeline for constructing patient DAGs}. Each numbered box corresponds to each step's LLM call and shows an example of the expected output.
Given an unstructured patient chart, the LLM conducts free-form clinical reasoning (step 1), creates a Gantt chart (step 2 creates nodes, step 3 creates edges), and finally encodes the graph into a parseable format (step 4 creates nodes, step 5 creates edges). 
LLM outputs are required to contain quotes from the patient chart, so that LLM reasoning is transparent and auditable.
The pipeline also includes checks to ensure validity of the final graph.
}
\label{fig:dag_construction}
\end{figure}

Central to this pipeline is designing step-by-step instructions to the LLM for deciding which nodes to add and when there were enough.
The primary goal is to include enough nodes for adequate confounder adjustment so that conditional exchangeability in Assumption~\ref{ass:dag_recovery} is satisfied.
To do this, prompts were based on the iterative graph expansion procedure of \citet{Guo2026-wl} and the disjunctive criterion \citep{VanderWeele2019-ws}, which starts from a graph containing only the exposure and outcome and expands it until causal identification can be achieved.
To help the LLM find these confounders, the prompts instruct it to walk backwards from the intervention targets (Type III nodes) and from the sink node, enumerating the events encountered along these paths.
A secondary goal was to include enough nodes in the DAG to improve the precision of timing estimates.
For this, the prompts encouraged including mediating events along the traced paths, as these shorten the gaps between consecutive nodes, and shorter gaps are easier to model because they have higher data support and lower variability (For instance, given the causal pathway ``initiation of IV antibiotics'' $\rightarrow$ ``end IV antibiotics'' $\rightarrow$ ``social worker visit'' $\rightarrow$ ``order a wheelchair'', there might be very few patients for whom we observe the initial and final event but likely many patients for whom we observe one or more consecutive pairs from this causal pathway).
Finally, because each additional node increases the LLM costs, prompts were designed to limit overly large DAGs.
The final prompts led to DAGs with an average of 14 nodes and 21 edges, which was manageable in terms of cost and query time.

The final pipeline used GPT-5-mini for generating the patient DAGs.
For prompt tuning, we used a small training dataset composed of a synthetic test patient and three real patients.
Outputs from the LLM pipeline were verified through a custom HTML interface that visualized the extracted graphs, the LLM's reasoning, and the raw patient timeline.
We developed a PHI-compliant annotation interface as well to evaluate agreement between our human experts and the LLM (see Appendix~\ref{sec:annotation_protocol} for details).

\subsubsection{Step 2}
\label{sec:step_2}

Step 2 of \textit{egg}-computation needs to sample timings for Type II and III nodes.
Below, we discuss how timing models for these nodes can be scaled up using LLMs while ensuring the timing estimates are aligned with the expert and real-world data distributions.

\paragraph{Type II nodes.}
To simulate timings of Type II nodes given their parent nodes, we consider fitting a pooled model that takes as input descriptions of the parent events, timings of the parent events, and a description of the child event.
The training data for such a pooled model is obtained by constructing the observed Gantt charts for patients in the training dataset and collecting their Type II nodes.
Pooling across patients is needed because a single hospitalization only contributes one draw from each waiting-time factor.
Since exact repeats of a conditioning event essentially never recur, the pooled model relies on a similarity assumption:
events with similar descriptions, parent timings, and patient context follow similar timing distributions (see Remark~\ref{rem:pooling} in the Appendix for the formal pooling assumption and its connection to Theorem~\ref{thm:correctness}).

We consider three ways to fit a model to this pooled dataset.
One option is to extract predefined tabular features from the free-form event descriptors and then train corresponding tabular models (e.g., random forests).
Another option is to train models that take in the original event descriptors directly, such as ridge regression with sentence embeddings as features.
A third option, based on in-context learning (ICL) \citep{Brown2020-qb}, prompts an LLM to use its reasoning abilities and prior knowledge to predict the most likely timing, given the timings observed in similar training examples.
We recommend choosing among these options by cross-validation on the pooled dataset.
In the experiments, we compare all three options under leave-one-out cross validation, with GPT-5 as the ICL option.

\paragraph{Type III nodes.}
There are generally two options for making an LLM output timestamp distributions that are aligned with the expert distribution $p_{\text{expert}}$ at a Type III node.
One option is to fine-tune the LLM or train a calibration model (e.g., using isotonic regression) if experts have annotated a sufficiently large number of samples, but this is infeasible for many hospital QI teams.
The other option, which we take here, is to specify the intervention in as much detail as possible such that the unknown expert-supplied distribution is precisely defined.
As demonstrated in prior work, ambiguity in the prompt is indeed one of the major drivers of misalignment between human and LLM outputs \citep{Tamkin2023-new, Subramonyam2024-at, Kothari2026-zo}.
Consequently, each intervention includes details on the motivation for the intervention, what the intervention is, examples of what the expected timings are under the intervention, and inclusion/exclusion criteria (see example prompts in Section~\ref{sec:type_i_prompt} of the Appendix).
After the initial prompt is created based on the detailed intervention, an iterative prompt-tuning procedure on a small set of patient DAGs can be used to further align the LLM-sampled event times with those from the human expert.
In practice, we suggest conducting additional sensitivity analyses that vary the direct timing impact of the intervention, either by perturbing the prompt or applying varying shrinkage factors on the LLM-anticipated effect.

\section{Experiments}

Having established the theoretical properties of egg-computation, we now study its empirical performance.
First, we study egg-computation assuming access to an oracle expert, using simulation studies to understand how it compares to classical causal inference methods.
Then on real-world patient timelines, we apply egg-computation with frontier LLMs to estimate the effects of candidate QI interventions on LOS and evaluate concordance of the LLM pipeline with human experts.

\subsection{Simulation studies comparing egg-computation to classical tabular causal inference methods}
\label{sec:simulation}

We have made two claims regarding the advantages of egg-computation over classical causal inference methods thus far.
First, egg-computation differs qualitatively in that it applies even in settings with limited or no data support, because it can additionally draw on expert input.
Second, egg-computation has a quantitative advantage in that its estimates can be more accurate through more faithful modeling of causal dependencies, in part because the method analyzes both structured and unstructured data.
While the former point is evident, the latter point has yet to be quantified and requires empirical study.

We present simulation results comparing egg-computation to existing methods for the following hypothetical clinical scenario (Figure~\ref{fig:sim_dag}).
After patient admission at time $T_{\text{Adm}} = 0$, this hypothetical hospital orders a diagnostic CT study at time $T_{O_{\text{CT}}}$ and generates a report at time $T_{R_{\text{CT}}}$; the same is done for an MRI study, with corresponding timestamps $T_{O_{\text{MRI}}}$ and $T_{R_{\text{MRI}}}$.
Depending on results from one or both of the studies, a final decision regarding the patient's treatment will be made at some time $T_{\text{Decision}}$.

Our goal is to understand the effect of a candidate intervention on the time to this final decision.
As shown in Figure~\ref{fig:sim_dag}, the causal dependencies between these event timings vary between patients due to patient baseline factors, in that (i) the CT report may require the MRI report coming out first or vice versa and (ii) the final decision may depend on one or both of these imaging reports.
For this simulation study, we generate $n=4{,}000$ patients whose edges are uniformly sampled at random; as such, half of the patients have the CT $\rightarrow$ MRI dependency while the other half have the MRI $\rightarrow$ CT dependency.
We study two candidate interventions, one that caps the time from CT order to report generation at 12 hours and another that caps the time to MRI report generation in the same fashion.
Each patient's sampled DAG is rendered as a synthetic clinical note that records a summary of the true causal structure. %

We compare estimation methods that vary in the DAGs used---either a single shared DAG (\texttt{Single-DAG}) versus extracted subject-specific DAGs (\texttt{DAG-extractor})---and how event timings are imputed---either by reusing the observed wait times as commonly done in qualitative analyses of Gantt charts (\texttt{Gantt-timing}) or by fitting a regression model for wait times (\texttt{OLS-timing}, \texttt{kNN-timing}) (see Appendix~\ref{sec:sim_details} for details).
Performance is measured in terms of estimation bias.

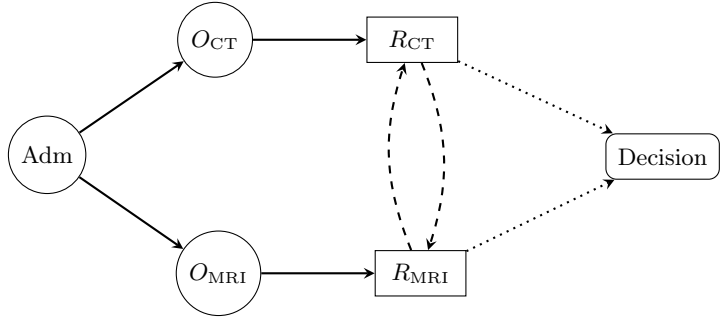
\begin{figure}[t]
\centering
\begin{tikzpicture}[
    node distance=1.5cm and 2cm,
    every node/.style={font=\small},
    order/.style={circle, draw, minimum size=0.8cm},
    report/.style={rectangle, draw, minimum width=1.2cm, minimum height=0.6cm},
    decision/.style={rectangle, draw, rounded corners, minimum width=1.5cm, minimum height=0.6cm},
    slack/.style={circle, draw, minimum size=0.6cm}
]
\node[order] (adm) {Adm};

\node[order, above right=0.8cm and 1.5cm of adm] (oct) {$O_{\text{CT}}$};
\node[order, below right=0.8cm and 1.5cm of adm] (omri) {$O_{\text{MRI}}$};

\node[report, right=1.5cm of oct] (ct) {$R_{\text{CT}}$};
\node[report, right=1.5cm of omri] (mri) {$R_{\text{MRI}}$};

\node[decision, right=2.5cm of $(ct)!0.5!(mri)$] (y) {$\text{Decision}$};
\draw[->, thick] (adm) -- (oct);
\draw[->, thick] (adm) -- (omri);
\draw[->, thick] (oct) -- (ct);
\draw[->, thick] (omri) -- (mri);
\draw[->, dotted, thick] (ct) -- (y);
\draw[->, dotted, thick] (mri) -- (y);
\draw[->, thick, dashed, bend left=20] (ct) to node[midway, right, xshift=2pt] {} (mri);
\draw[->, thick, dashed, bend left=20] (mri) to (ct);

\end{tikzpicture}
\caption{Causal graph generation scheme for the simulation study.
Each patient's causal graph has all solid edges, only one of the dashed edges, and one or both of the dotted edges.
}
\label{fig:sim_dag}
\end{figure}

\begin{table}[t]
\centering
\caption{Simulation results.
The signed bias (hours) of the estimated average time saved, $\hat{\tau} - \tau$, under each intervention, reported as mean $\pm$ standard error over the 50 replications.
A positive value overstates the speed-up the intervention delivers and a negative value understates it.
Each method pairs a DAG (patient-specific \texttt{DAG-extractor} versus a shared \texttt{Single-DAG} with an assumed report ordering) with a timing update (\texttt{Gantt-timing} reuses the factual time differences exactly; \texttt{OLS-timing} and \texttt{kNN-timing} fit models for the timing gaps).
\texttt{DAG-extractor + Gantt-timing} is the oracle configuration in this simulation, since the patient-specific DAG is correct and exact propagation matches the data-generating process.
The Priority-CT and Priority-MRI columns are the two interventions.
}
\label{tab:sim_results}
\begin{tabular}{lcc}
\toprule
& \multicolumn{2}{c}{ATE Bias (h)}\\
\cmidrule(lr){2-3}
Method & Priority-CT & Priority-MRI \\
\midrule
\texttt{DAG-extractor + Gantt-timing} & $0.00 \pm 0.00$ & $0.00 \pm 0.00$ \\
\texttt{DAG-extractor + OLS-timing} & $-0.35 \pm 0.02$ & $-0.22 \pm 0.02$ \\
\texttt{DAG-extractor + kNN-timing} & $2.75 \pm 0.01$ & $2.11 \pm 0.01$ \\
\midrule
\texttt{Single-DAG + Gantt-timing} (CT$\to$MRI) & $13.01 \pm 0.06$ & $-2.56 \pm 0.07$ \\
\texttt{Single-DAG + Gantt-timing} (MRI$\to$CT) & $-5.37 \pm 0.04$ & $14.19 \pm 0.07$ \\
\texttt{Single-DAG + OLS-timing} (CT$\to$MRI) & $10.73 \pm 0.08$ & $-3.05 \pm 0.16$ \\
\texttt{Single-DAG + OLS-timing} (MRI$\to$CT) & $7.06 \pm 0.11$ & $5.47 \pm 0.16$ \\
\bottomrule
\end{tabular}

\end{table}

Unsurprisingly, Table~\ref{tab:sim_results} shows that \texttt{DAG-extractor + Gantt-timing} achieves zero bias, since \texttt{DAG-extractor} recovers the true patient-specific structure and \texttt{Gantt-timing} uses the true timing model.
The worst performing methods are those that rely on \texttt{Single-DAG}, where the shared DAG is always wrong for a majority of the cohort; as a result, the estimation bias can reach up to 14 hours.
In contrast, imperfect timing models with the patient-specific DAG extractor added comparatively little bias: at most 0.4 hours for \texttt{OLS-timing} and 2.8 hours for \texttt{kNN-timing}.
Thus, this simulation demonstrates that when patient-specific causal structure is heterogeneous and latent, methods that impose a single population-level DAG incur large and unpredictable biases.
Egg-computation can eliminate structural misspecification and achieve unbiased estimation when it can accurately recover the subject-specific DAG auxiliary data such as clinical notes.

\subsection{Real-world Evaluation of Candidate QI interventions}
\label{sec:results}

Building on prior LLM-based analyses of major bottlenecks and delays at the Zuckerberg San Francisco General Hospital (ZSFG) \citep{Vossler2026-ew}, we analyze eleven candidate QI interventions that the hospital is considering in response.
Described in more detail in Appendix~\ref{sec:intervention_definitions}, candidate interventions range from accelerating completion of high-priority orders, advancing the initiation of certain processes, and increasing weekend services.
Here we present quantitative and qualitative evaluations of the egg-computation pipeline on this real-world dataset.

Using the same data pipeline as \citet{Vossler2026-ew}, we extracted 2{,}193 inpatient timelines composed of clinical notes, clinical orders and events, and laboratory/test results for adult hospitalizations from the five most common diagnoses. %
Applying the LLM-based egg-computation pipeline, we first screened each hospitalization for whether it was eligible for an intervention and then  estimated the causal effect on average LOS for the filtered population.
To control computation costs, egg-computation pipeline was run for at most 30 eligible hospitalizations per diagnosis group per intervention, resulting in a total of 1{,}531 LLM-generated DAGs.
The experiment's LLM API cost totaled \$490; a cost breakdown is provided in Appendix~\ref{sec:computational_cost}.

\subsubsection{Expert agreement with LLM-constructed DAGs}
For LLMs to successfully scale up egg-computation, human experts should generally agree with DAG generations from the LLM pipeline.
To evaluate this, five experts from the hospital's QI team (PV, JO, AH, LZ, JF) were asked to evaluate LLM-generated DAGs for two candidate interventions: early discharge planning and disposition-critical imaging studies. %
We evaluated concordance between DAG edges extracted by the LLM pipeline versus those indicated by expert annotators, to assess the reliability of using LLMs for scaling up human experts. %

Twelve hospitalizations were annotated per intervention, so a total of 24 intervention-hospitalization pairs were annotated, where 8 were double-annotated to measure inter-annotator agreement.
Given that patient timelines often spanned multiple days and averaged roughly 16{,}000 words, we designed a structured annotation workflow and interface that would balance time/resource constraints against potential biases in the annotation process.
Our final annotation workflow is based on \citet{Vossler2026-ew}: we asked annotators to check which extracted nodes should be direct children of the intervention $A$ (i.e., check which nodes are Type III) and which edges should be included for nodes downstream of the intervention (i.e., check input edges to all Type II nodes); see the full protocol in Appendix~\ref{sec:annotation_protocol}.
A total of 476 node verdicts and 243 edge verdicts were collected.

The annotation study revealed high concordance between the nodes and edges generated by the LLM pipeline and the human annotators (Table~\ref{tab:llm_is_ok}).
Across both interventions, specificity and precision were consistently high.
The annotators mostly agreed with the structure the LLM included.
Recall was also high but generally lower than specificity and precision.
This is the expected, as the LLM is prompted to only include nodes and edges it can confidently justify from the timeline (Section~\ref{sec:step_1}).
Inter-annotator agreement was also high (Table~\ref{tab:interannotator} of the Appendix), with an average agreement rate of over 80\%.
Most importantly, we re-ran egg-computation on LLM-generated versus annotator-corrected DAGs (Table~\ref{tab:llm_is_ok}, final row). 
The mean difference in the estimated average time saved was $-2.6$ hours for early discharge planning and $8.1$ hours for disposition-critical imaging, with 19 of the 24 pairs yielding identical estimates, demonstrating how LLM-based estimates are indeed quite close to human-based estimates.

\begin{table}
    \centering
    \vspace{-0.3cm}
    \caption{Concordance between the LLM-generated and annotator-corrected DAGs, with 95\% confidence intervals.
    For intervention targets, we report the share of events the annotator judged non-targets that the LLM also left unflagged (specificity) and the share of annotator-endorsed targets the LLM endorsed (recall).
    For input edges to nodes downstream of intervention, we report the proportion of LLM-proposed edges that were endorsed (precision) and the proportion of annotator-endorsed edges present in the LLM's graph (recall).
    The final row reports the mean difference (hours) in the estimated average time saved when egg-computation was rerun on LLM-generated versus the annotator-corrected DAGs.
    }
    \label{tab:llm_is_ok}
    \resizebox{\linewidth}{!}{%
    \begin{tabular}{lcc}
    \toprule
     & Early discharge planning & Disposition-critical imaging \\
    \midrule
    Intervention-target specificity (\%) & 96.5 [91.9, 100.0] & 96.5 [91.4, 100.0] \\
    Intervention-target recall (\%) & 88.9 [79.2, 100.0] & 95.7 [85.2, 100.0] \\
    Edge precision (\%) & 91.8 [85.6, 97.7] & 94.2 [89.2, 100.0] \\
    Edge recall (\%) & 80.2 [65.6, 96.4] & 86.0 [73.5, 100.0] \\
    \midrule
    Mean difference in est.\ time saved (h) & -2.6 [-8.3, 0.8] & 8.1 [0.0, 24.2] \\
    \bottomrule
    \end{tabular}%
    }
    \vspace{-0.3cm}
\end{table}

\subsubsection{Accuracy of timing models for downstream events}

For egg-computation to produce accurate causal estimates, the timing models for downstream events should approximate the true conditional distributions $p\left(W_{k} \mid T_{\Pa{k}}, X,D\right)$ of Type II nodes.
Under leave-one-out cross-validation at the hospitalization level, we evaluated the timing models of Section~\ref{sec:step_2} on all 4{,}790 downstream node timings across the 736 hospitalizations described above, pooled over the five diagnosis groups and the eleven candidate interventions.
The comparison includes two baselines (predict zero delay and predict the median delay among event pairs of the same type), three regression models (one of which includes a sentence embeddings), and a simple ICL solution with GPT-5 that requires no feature engineering nor model fitting.
Experiment details are in Appendix~\ref{sec:type_iii_appendix}.

As shown in Table~\ref{tab:type_iii_results}, ICL attains the lowest mean absolute error (24.1 hours against 27.2 for the closest baseline).
The performance of ICL is significantly better than all five baselines for eight of the eleven interventions; the performance differences are not statistically significant for the remaining three (Table~\ref{tab:type_iii_by_intervention}, final column).
The latter three cases should be taken into context: the structure of their wait time distributions are sharply zero-inflated and have 48--71\% of their wait times under one hour.
As such, the trivial baseline of predicting zero can perform very well for these special cases.
The fact that the ICL method outperformed most methods despite its simplicity mirrors recent results finding that few-shot LLM prompting can be highly competitive with trained tabular models \citep{Hegselmann2022-hl}, particularly in settings with highly complex data and small sample sizes.
Interestingly, closer analysis revealed that ICL was better at predicting longer waiting times than regression models (Appendix~\ref{sec:type_iii_appendix}), as longer wait times tend to be noisier and have less data support.

\begin{table}[htb]
\caption{Wait time prediction accuracy for Type II nodes (4{,}790 downstream node timings from 736 hospitalizations). Mean Absolute Error (MAE) are shown in hours, along with Relative MAE compared to that of the zero-gap baseline.
Values below one improve on predicting zero delay for every node.
The final column gives each baseline's excess MAE over ICL, with 95\% intervals from a cluster bootstrap over hospitalizations.
Per-intervention and per-delay-stratum breakdowns are given in Tables~\ref{tab:type_iii_by_intervention} and~\ref{tab:type_iii_by_stratum}.}
    \label{tab:type_iii_results}
    \centering
    \begin{tabular}{lccc}
    \toprule
    Method & MAE & Rel.\ MAE & $\Delta$MAE vs.\ ICL [95\% CI] \\
    \midrule
    Zero gap          & 28.30 & 1.00 & +4.18 [3.57, 4.77] \\
    Pair-type median  & 27.23 & 0.96 & +3.11 [2.57, 3.67] \\
    Ridge             & 34.30 & 1.21 & +10.18 [9.37, 10.99] \\
    Random forest     & 36.58 & 1.29 & +12.46 [11.18, 13.90] \\
    Embedding + ridge & 33.79 & 1.19 & +9.67 [8.57, 10.70] \\
    ICL (GPT-5)       & \textbf{24.12} & \textbf{0.85} & --- \\
    \bottomrule
    \end{tabular}
\end{table}

\subsubsection{Counterfactual Gantt charts and estimated rankings from egg-computation}

\begin{figure}
    \centering
    \includegraphics[width=0.9\linewidth]{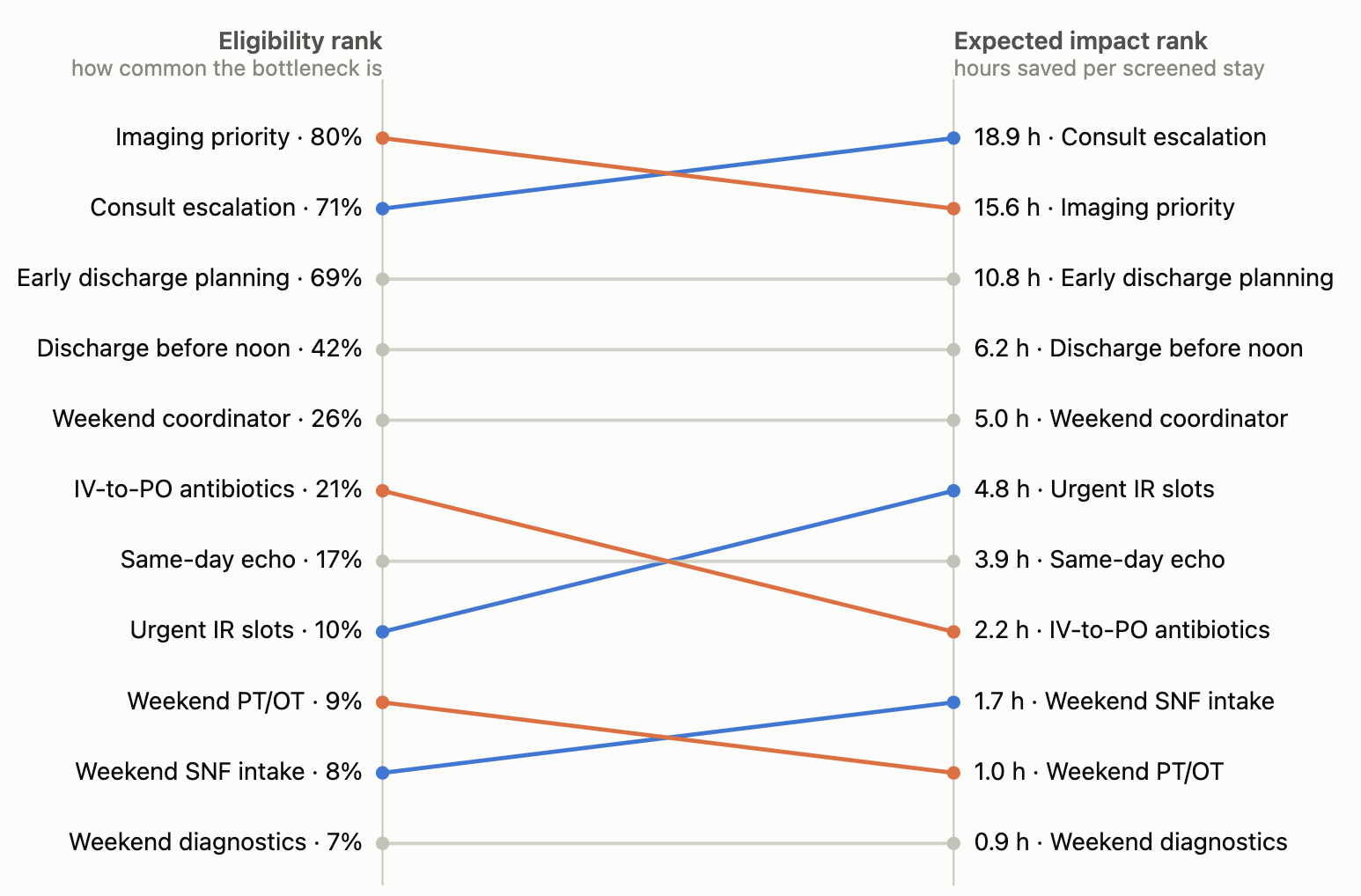}
    \vspace{-0.2cm}
    \caption{
    Prevalence-based versus causal prioritization of the eleven candidate interventions.
Eligibility rate is the share of hospitalizations that meet an intervention's inclusion criteria.
Expected time saved is the causal effect estimate from egg-computation.
Interventions are ordered by expected impact.
    }
    \label{fig:prevalence_vs_savings}
\end{figure}

Having evaluated the accuracy of different components of our LLM-powered egg-computation pipeline, we now apply the final pipeline to obtain expert-informed causal effect estimates for all eleven candidate interventions.
To illustrate results from egg-computation, Figure~\ref{fig:example_patient} shows an analysis of a randomly sampled hospitalization.
Here, an ischemic stroke patient is waiting on a transthoracic echocardiogram (TTE), a common requirement prior to discharge.
The TTE was ordered eight hours after admission but the imaging procedure was not completed until hour 72 of the patient's stay.
The clinician's disposition review followed within an hour of the TTE report.
Under the disposition-critical imaging priority intervention, egg-computation shifts the TTE completion and report roughly two days earlier, copies wait times for Type I nodes, and re-predicts timings for events downstream of the intervention targets.
The imputed LOS is consequently 42 hours shorter.

Using egg-computation, we obtain estimated rankings for the candidate interventions by their expected impact, i.e., the eligibility rate multiplied by the average time savings (Figure~\ref{fig:prevalence_vs_savings}, Table~\ref{tab:prevalence_vs_savings}).
Appendix~\ref{sec:savings_appendix} provides sensitivity analyses of the intervention rankings.
To understand the face validity of these ranking estimates, we asked four clinicians to review these results through a structured interface (Figure~\ref{fig:example_task_2}, protocol in Appendix~\ref{sec:clinician_review}).

\begin{figure}
    \centering
    \includegraphics[width=\linewidth]{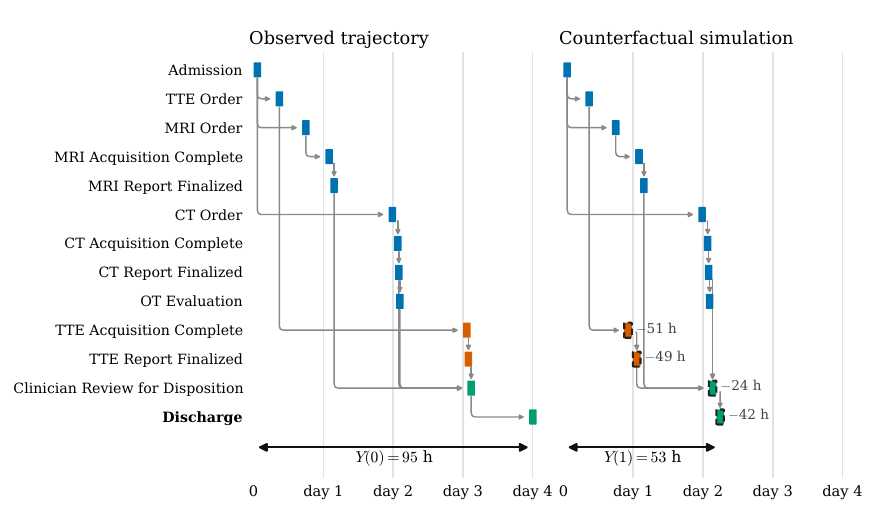}
    \definecolor{typeI}{HTML}{D55E00}
    \definecolor{typeII}{HTML}{0072B2}
    \definecolor{typeIII}{HTML}{009E73}
    \newcommand{\typeswatch}[1]{\protect\raisebox{-0.2ex}{\protect\tikz{\protect\node[rectangle, rounded corners=1pt, fill=#1, minimum height=2.2mm, minimum width=5mm, inner sep=0pt]{};}}}
    \vspace{-1.3cm}
    \caption{
    \textbf{Observed and counterfactual trajectories for one hospitalization} (ischemic stroke and disposition-critical imaging priority). 
    Each row is one event in the patient's Gantt chart, placed at its event time in hours from admission, and gray arrows are the patient's DAG edges.
    Bar colors give each event's relation to the intervention: children of $A$, i.e., the intervention targets (\typeswatch{typeI}), non-descendants of $A$ (\typeswatch{typeII}), and descendants of $A$ that are not children of $A$ (\typeswatch{typeIII}).
    In the counterfactual panel, dashed outlines mark events whose timing changed and are labeled with the shift in hours.
    }
    \label{fig:example_patient}
\end{figure}

Ranking by expected impact, ``Consult response escalation'' ranks first at 18.9 hours per screened hospitalization, followed by ``Disposition-critical imaging priority'' at 15.6.
While CI for their estimated impacts are wide, the resulting shortlist is stable, with the same three interventions occupying the top three positions in 89\% of the bootstrap resamples (Appendix~\ref{sec:rank_stability}).
Critically, these causal rankings differ from the simpler option of ranking interventions solely based on the percent of patients who are eligible.
``Disposition-critical imaging priority'' is ranked as having the most eligible patients, while ``Consult response escalation'' is estimated to have a higher causal impact.
When clinicians were asked to order the top three interventions based on their experiences at the hospital (same across eligibility and egg-computation rankings), we found the egg-computation estimates to align more with clinical intuition: three of the four clinicians ranked ``Consult response escalation'' above ``Disposition-critical imaging priority'' (the pair on which the two rankings disagree) and two exactly reproduced the egg-computation ranking, while none reproduced the eligibility ranking.
Using the estimated causal effect also swaps the positions of ``Urgent IR procedure block'' (a protected block for urgent interventional-radiology procedures) and ``Antimicrobial stewardship'' (which encourages earlier conversion from intravenous to oral antibiotics), moving the former up two positions relative to the eligibility ranking.
Its estimated impact is more than double that of ``Antimicrobial stewardship'', even though its eligibility rate is less than half.
When asked about this, clinicians agreed egg-computation's ordering was more plausible, as a prolonged IV antibiotic course usually reflects a deliberate clinical decision and oral antibiotics would not be given substantially earlier through the implementation of an automatic IV-to-oral notification system.

\begin{figure}
  \centering
  \includegraphics[width=\textwidth]{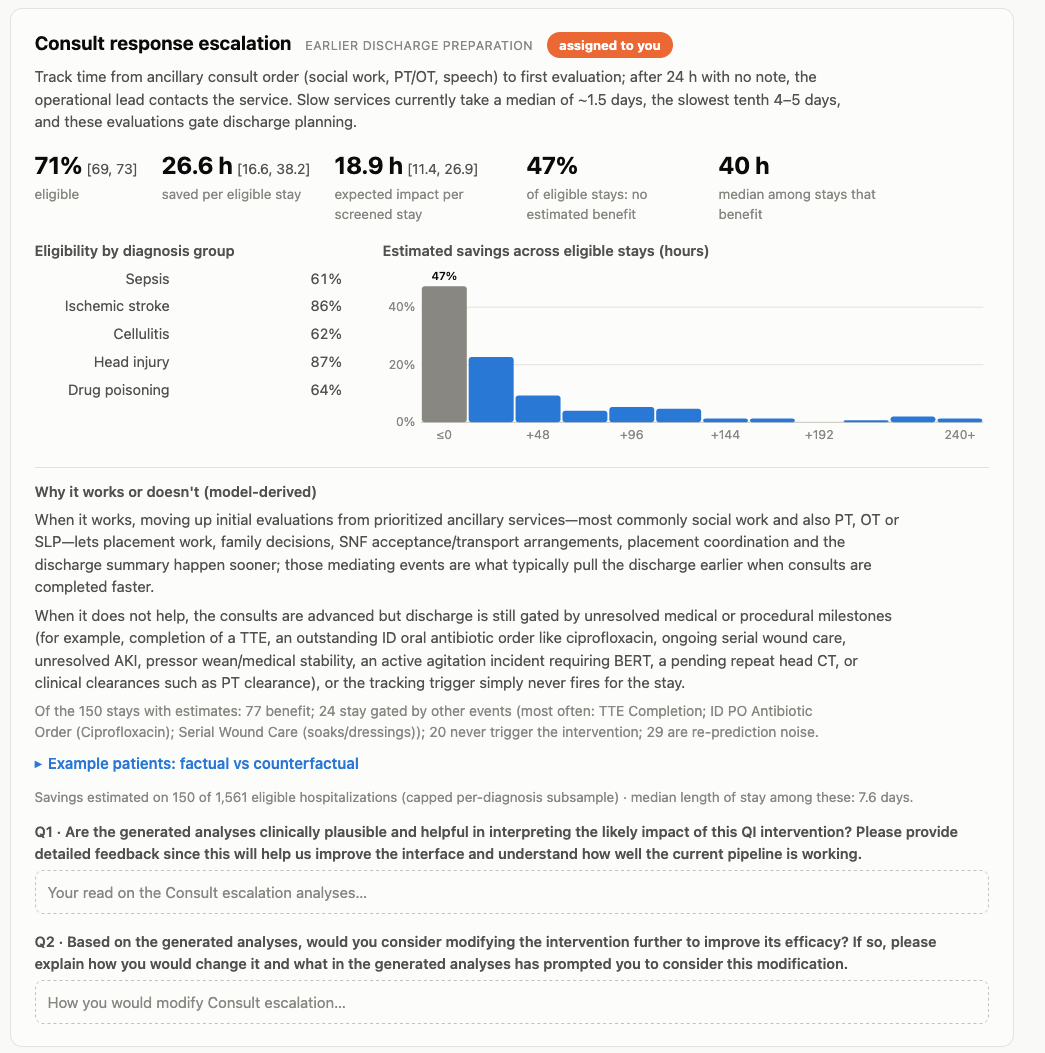}
  \caption{\textbf{LLM summary of egg-computation output of a candidate intervention provided for clinician review.}
  For each intervention, reviewers were shown the intervention description, the eligibility rate and estimated savings, the distribution of estimated savings across eligible hospitalizations, and an LLM-generated summary, based on all the hospitalizations, of when the intervention did and did not save time. An expandable section provided patient-level examples comparing each observed timeline to its counterfactual.
  Reviewers were then asked whether the analyses were clinically plausible and whether they would modify the intervention's design.
  }
  \label{fig:example_task_2}
\end{figure}

Egg-computation results also revealed additional insights into intervention design.
For instance, while the third-ranked intervention of early discharge planning had a comparably high eligibility rate (69\%, versus 80\% for the second-ranked intervention), its causal estimate was much lower (around two-thirds of the second highest causal effect).
This estimate was lower than the literature would suggest, as other hospitals have reported that early discharge planning substantially accelerated their workflows \citep{Wassef2018-kg}.
Investigating the egg-computation results further, clinicians agreed with the LLM's analysis that early discharge planning was often blocked by patient needs (e.g., medical needs) that a social worker or case manager could not resolve; in fact, clinicians felt that the LLM was often still too lenient in assuming the social worker or case manager could resolve patient needs early.
Based on the results, clinicians hypothesized that realizing the purported effectiveness of this intervention likely required much deeper organizational changes than those originally described in the intervention's specification, perhaps involving additional hospital staff (e.g., nursing) \citep{Bechir2025-tg, Bryant-Lukosius2015-qq}.

Overall, clinicians agreed that the egg-computation analysis was helpful, and they singled out the patient-level counterfactual analyses like those in Figure~\ref{fig:example_patient} as particularly helpful.
Being able to easily review a hospitalization, what could have happened under the intervention, and how many hours would have been saved made it easier to judge when an intervention would or would not be applicable.
By providing the counterfactual Gantt chart visualizations, clinicians could readily audit plausibility of the overall procedure, gain trust its causal reasoning, and provide feedback on how the LLM-assisted system could be further improved.

\vspace{-0.2cm}
\section{Discussion} \label{sec:discuss}
This work introduces a causal framework for estimating the \textit{average time saved} by a candidate intervention that has not yet been implemented; this causal question commonly arises in hospital QI but in other contexts as well.
The key insight underlying the proposal is that expert knowledge can be combined with data-driven modeling to obtain accurate causal estimates \textit{if we carefully elicit information from the right sources}: experts should supply information in areas where they are accurate but data support is limited, while data should supply information in areas where modeling is accurate but experts are not.
Using the language of Gantt charts as the bridge between expert reasoning and formal causal reasoning, we formalize a model for Gantt charts with both probabilistic and causal semantics and then develop expert-guided g-computation, or egg.
Finally, we develop a LLM-assisted implementation to scale expert reasoning required by the framework. 
Simulations demonstrate improved effect estimation relative to conventional causal inference methods, while the empirical study shows that the LLM-assisted pipeline is a practical tool whose outputs are largely concordant with those of human experts.

The current formulation suggests several natural extensions.
In particular, operational interventions often rely on shared resources, so accelerating one patient’s care may affect another patient’s trajectory.
Our current estimand should therefore either be interpreted as the effect for patients in the same population strata defined by $X$ keeping the same or effect of the intervention assuming unlimited resources for the proposed intervention and thus an upper-bound.
One solution is to apply egg-computation jointly to batches of patients instead, with each patient's access to an intervention or resource determined by concurrent demand from other patients.
Formal inferential procedures are also needed to construct confidence intervals that account for sampling variability and uncertainty in expert-supplied inputs.
More extensive calibration with expert judgment will also further strengthen its performance. 

To our knowledge, this is the first framework to bring a causal perspective to Gantt-based process improvement, providing an estimation strategy that combines expert knowledge with observed data and an LLM-assisted approach for scaling its implementation.
More broadly, egg-computation can be adapted to other time-ordered processes with task dependencies represented by Gantt charts, such as production scheduling in manufacturing, software development planning, and human resource management. This broad applicability positions egg-computation to become a useful framework for evaluating and guiding operational improvements across complex systems.

\makeatletter
\@oupdraftfalse
\makeatother
\normalsize
\bibliographystyle{plainnat}
\bibliography{paperpile, addtl_refs}

\appendix

\section{Proofs for egg-computation}
\label{sec:proofs}

The following lemmas, established for Type I and II nodes respectively, are used to prove \cref{thm:correctness}. 
In this Appendix, we use the shorthand $[k]:=\{0,\dots,k\}$ with $[-1] := \emptyset$.

\begin{lemma} \label{lem:nd}
Under \cref{assump:consistency,assump:exo,assump:exclusion,assump:local-markov},
for $k \in \Nd{A}$, we have $T_k(1) = T_k$.
\end{lemma}

\begin{proof}[Proof of \cref{lem:nd}]
Observe that $\Nd{A}$ is ancestral in the sense that $j \in \Nd{A}$ and  $l \in \Pa{j}$ imply $l \in \Nd{A}$.
In light of this, the result follows from the first case of \cref{assump:exclusion} and \cref{eq:T-W-a}.
\end{proof}

\begin{lemma} \label{lem:de}
Under \cref{assump:consistency,assump:exo,assump:exclusion,assump:local-markov}, for $k \in \De{A} \setminus \Ch{A}$, we have the equality in conditional distributions:
\begin{multline*}
p\left(W_k(1) = w_k \mid T_{[k-1]}(1) = t_{[k-1]}, X, \sD \right) =p\left(W_k = w_k \mid T_{\Pa{k}} = t_{\Pa{k}}, X, \sD \right), \\ w_k \in [0,+\infty], \quad t \in [0,+\infty]^{K+1}.
\end{multline*}
\end{lemma}

\begin{proof}[Proof of \cref{lem:de}]
For any $w_k$ and $t$, we have
\begin{align*}
& \quad p\left(W_k(1) = w_k \mid T_{[k-1]}(1) = t_{[k-1]} \right) \\
\text{(by \cref{assump:local-markov})} &= p\left(W_k(1) = w_k \mid T_{\Pa{k}}(1) = t_{\Pa{k}} \right) \\
\text{(by 2nd case of \cref{assump:exclusion})} &= p\left(W_k(0) = w_k \mid T_{\Pa{k}}(0) = t_{\Pa{k}} \right) \\
\text{(by \cref{assump:exo})} &= p\left(W_k(0) = w_k \mid T_{\Pa{k}}(0) = t_{\Pa{k}}, A=0 \right) \\
\text{(by \cref{assump:consistency})} &= p\left(W_k = w_k \mid T_{\Pa{k}} = t_{\Pa{k}} \right).
\end{align*}
\end{proof}

\begin{remark}[Pooling assumption for the Type II timing models] \label{rem:pooling}
\cref{lem:de} reduces the counterfactual waiting-time distribution of a Type II node to the factual factor $p\left(W_k \mid T_{\Pa{k}} \right)$, but estimating this factor from data raises two issues that the pooled model of Section~\ref{sec:step_2} is designed to address.
First, a single hospitalization contributes only one draw from each factor, so any estimator must pool information across patients.
We therefore assume that hospitalizations whose child event and parent configuration match are drawn from the same waiting-time distribution, which makes the factors identifiable from a population where each configuration occurs.
Exact configuration matches are essentially non-existent in practice, since event descriptions are free text, sets of parent nodes differ from patient to patient, and parent timings are continuous.
As a result, the pooled model uses information from hospitalizations that are similar but not identical, assuming that events with similar descriptions, parent timings, and patient context follow similar timing distributions.
Whereas tabular and embedding models fix a notion of similarity in advance through their features, an LLM can weigh similarity through the clinical context of the descriptions themselves, which motivates the in-context learning option of Section~\ref{sec:step_2}.
Theorem~\ref{thm:correctness} requires the Type II timing models to equal the true conditional distributions, and in our implementation that condition holds to the extent this similarity assumption does.
Second, pooling changes what the model must condition on.
The per-patient statements of Section~\ref{sec:causal-wait} hold conditionally on the baseline covariates $X$, so \cref{lem:de} suppresses $X$ from the notation: the DAG is constructed from the patient's own chart, Type I node timings are copied from the same trajectory, and Type III node distributions are specified for the intervention acting on that patient, so none of these steps combines information across patients and each conditions on $X$ implicitly.
A model pooled across patients loses this implicit conditioning, and because the waiting time between parent and child events may depend on patient characteristics, it must condition on $X$ explicitly, i.e., estimate $p\left(W_k \mid T_{\Pa{k}}, X\right)$.
\end{remark}

\begin{lemma}[g-computation] \label{thm:g}
For $k=0,1,\dots,K$, let $T_k'(1)$ be iteratively computed as follows:
\begin{enumerate}
\item if $k \in \Nd{A}$, let $T_k'(1) = T_k$;
\item if $k \in \Ch{A}$, draw $W_k'(1) \sim  p\left(W_k(1) \mid T_{\Pa{k}}(1) =T_{\Pa{k}}'(1), X, \sD \right)$ and compute $T_k'(1)$ via \eqref{eq:T-W-a};
\item if $k \in \De{A} \setminus \Ch{A}$, draw $W_k'(1) \sim p\left(W_k \mid T_{\Pa{k}} = T_{\Pa{k}}'(1), X, \sD \right)$ and compute $T_k'(1)$ via \eqref{eq:T-W-a}.
\end{enumerate}
Then, we have $T_K'(1) =_{d} Y(1)$.
\end{lemma}

\begin{proof}[Proof of \cref{thm:g}]
By marginalization, it suffices to show 
\[ (T_0'(1),T_1'(1),\dots,T_K'(1)) \mid X, \sD =_{d} (T_0(1),T_1(1),\dots,T_K(1)) \mid X, \sD. \]
This identity follows from applying \cref{lem:nd,lem:de} and \cref{assump:local-markov} to each corresponding term in the factorization
\begin{multline} \label{eq:factorize-g}
p(T_0(1),T_1(1),\dots,T_K(1) \mid X, \sD) = p\left(T_{\Nd{A}}(1) \mid X, \sD \right) \\
\times \prod_{k \in \Ch{A}} p\left(T_k(1) \mid T_{[k-1]}(1) \mid X, \sD \right) \prod_{k \in \De{A} \setminus \Ch{A}} p\left(T_k(1) \mid T_{[k-1]}(1) \mid X, \sD \right).
\end{multline}
\end{proof}

\noindent We are ready to prove our main result.
\begin{proof}[Proof of \cref{thm:correctness}]
Given the additional \cref{ass:dag_recovery,assump:timing_model,assump:expert}  about the expert's specification, the first statement follows from \cref{thm:g}. 
Then, by linearity of expectation, we have
\[ \E \hat{\tau} = \E Y - \E T_K'(1) = \E Y(0) - \E Y(1) = \tau, \]
where we used the fact that $Y = Y(0)$ by \cref{assump:consistency}.
\end{proof}

\section{Related work}
\label{sec:related}

\textbf{Causal Inference Methods and their Assumptions.}
Causal inference methods for longitudinal data allow treatment effects to be estimated over long-ranging time sequences, such as using g-computation methods for discrete time grids \citep{Robins1986-kr} and Local Independence Graphs for continuous time \citep{Didelez2008-fw, mogensen2020markov, roysland2025graphical, Wald2024-yp}.
Nevertheless, existing methods largely rely on the core assumptions of conditional exchangeability, consistency, and positivity.
Furthermore, existing methods typically assume the intervention's effect is encapsulated by a simple formula/function.
However, complex interventions, such as those described via text, may not have such simple descriptors.
Finally, existing methods typically assume a common causal graph, whereas the causal graph may not be constant over time \citep{Keshavarz2020-de}.

\textbf{Role of Experts in Causal Inference.}
In traditional causal inference methods, experts play the role of specifying the causal structure, as causal structure cannot be determined from data alone.
To this end, various guidelines have emerged regarding what variables to include in the causal graph \citep{VanderWeele2019-ws, Guo2026-wl}.
Nevertheless, traditional applications rely on human judgment, which makes these methods difficult to scale.
Beyond specifying causal structure, experts are also relied upon to supply quantitative judgments, such as eliciting effect magnitudes as probability distributions \citep{OHagan2006-new} and bounding the plausible strength of unmeasured confounding in sensitivity analyses \citep{CinelliHazlett2020-new}.

\textbf{LLMs for Causal Discovery and Mimicking Expert Reasoning.}
Given the technical and practical difficulties of causal DAG construction, there is now increasing interest and evidence that LLMs can assist in conducting causal reasoning \citep{Liu2025-uc}.
Within causal inference/discovery, prior works have demonstrated that LLMs can accurately infer pairwise causal relationships simply through variable names \citep{Kiciman2024-gz, Liu2024-kq}.
LLMs have also achieved good results when inferring larger multi-node causal graphs from scientific/medical abstracts and news articles (e.g., through iterated pairwise queries) \citep{Antonucci2024-xb, Gendron2025-cj}.
Most importantly, with proper calibration, LLMs can accurately mimic human reasoning \citep{Argyle2023-yl, Aher2022-sd}, which motivates our approach in using LLMs to conduct \textit{expert}-guided causal reasoning.
We do not claim to perform \textit{exact} causal reasoning;  because the graph itself is a reflection of \textit{expert} judgment, this mirrors standard practice in causal inference.
To our knowledge, the closest work on using LLMs to conduct end-to-end causal inference is \citep{Cotta2024-qc}, which considers a different setting where the goal is to scrape internet data to estimate the causal effect by (i) extracting the outcome, treatment, and confounders and (ii) running classical causal inference methods such as IPW.

\textbf{Informatics Approaches to Healthcare Flow Modeling.}
Parallel to the statistical literature, informatics and operations research communities have long examined hospital flow through process-improvement frameworks such as Lean methodology \citep{Catalyst2018-ho, D-Andreamatteo2015-zg}, Gantt charts \citep{Wilson2003-zq, Clark1922-ih},  discrete-event simulation (DES) \citep{Banks2005-bx}, and queueing theory \citep{Shortle2018-gc, Green2006-yl}.
Nevertheless, existing methods lack rigorous causal identification guarantees and are not designed to leverage unstructured data to improve the causal reasoning.
This work illustrates how the operational practice of modeling with Gantt charts can be combined with formal causal reasoning methods (e.g., g-computation variants) to produce causal estimates with explicit identification assumptions.

\section{The LLM pipeline for DAG construction}
\label{sec:dag_pipeline}

This section describes what each of the pipeline's five steps instructs the LLM to produce, how the graphs are kept compact, and the checks applied to the output.
Each admission's notes, orders, and results are merged into a single timeline with repeated note text deduplicated; these timelines have a median length of roughly 48{,}000 tokens.
The first step reads this full patient timeline and the operational intervention and, in free-form clinical prose, identifies the events that determine length of stay.
It enumerates the intervention targets (Type~III nodes) that the intervention directly re-times, the sink-gating events that must complete before discharge can occur, and, tracing backward from both toward admission, the intermediate events and candidate confounders that lie along those causal paths.
The confounder search is carried out in this first step, following the iterative graph-expansion view of \citet{Guo2026-wl} and the disjunctive cause criterion of \citet{VanderWeele2019-ws}, so a pre-exposure variable is retained when it is a cause of an intervention target, a discharge-blocking event, or both.
The second and third steps turn this reasoning into the Gantt chart's fixed set of events and then the timing dependencies among those events.
The second step finalizes the event set by deduplicating and filtering the first step's events, giving each event a start and an end timestamp read from the chart, and tagging which events the intervention directly affects.
The third step holds the event set fixed and records, for every event, which other events' timings determine when it starts and, for events with nonzero duration, when it ends.
The fourth and fifth steps re-express the same content in the structured graph format that the downstream code consumes, again splitting the work into nodes first and edges second.
The fourth step emits the graph's nodes with their clinical category, event type, and timestamps, while the fifth step emits the directed edges, each carrying a one-sentence causal rationale and a verbatim supporting quote from the timeline.
Consistent with the design principle stated in the main text, the earlier steps are prompted to spend their reasoning on clinical and causal judgment and the later steps to spend it on faithful transcription into the required format.
Algorithm~\ref{algo:dag_construction} summarizes the recipe the prompts encode.

No step sets an explicit target for the number of nodes or edges; the reported average of roughly 14 nodes and 21 edges follows from two instructions rather than a fixed quota.
The first is a granularity target applied during backward tracing: the first step is told to ``aim for a granularity such that all causal links span no more than'' a set number of hours (twenty-four in our runs), with any longer link ``split further by adding more detailed event breakdowns.''
This adds nodes where consecutive events are far apart in time, the gaps that are hardest to model and have the least data support.
The second is the stopping rule implied by the disjunctive cause criterion, under which backward expansion halts once every retained pre-exposure variable is a cause of an intervention target or a sink-gating event, bounding the confounder set rather than requiring an enumeration of all baseline covariates.
The second step then removes near-duplicate events and drops any event that does not help determine the timing of another event, trimming the graph before it is formalized in the structured steps.

Once the structured node and edge steps have run, each graph passes through a set of deterministic checks before it is accepted.
The extracted timestamps are compared against the source timeline, and any timestamp with no charted event within an hour is flagged for later review.
The node and edge text is then parsed into a graph, and a parse failure or a structural problem triggers one automatic repair attempt in which the offending output is returned to the model together with the specific errors found and a request to fix only those.
Structural validation checks for cycles using a topological sort over the split start and end endpoints of each event and confirms that every event has a directed path to the discharge node.
Small temporal inconsistencies in which a parent's timestamp slightly postdates its child's are corrected automatically within a few-minute tolerance, while larger violations are recorded as issues.
Events that have no cause or no downstream effect are pruned, with the exception of the reserved admission, baseline, and discharge nodes and the intervention targets, which are never dropped.
Because the baseline nodes are connected to every non-baseline event automatically during post-processing, the LLM is instructed to omit those edges from its output.
Graphs that still fail structural validation after the single repair attempt are kept and marked rather than discarded; the sensitivity of the timing results to these graphs is reported in Section~\ref{sec:type_iii_validity}.

\begin{algorithm}
\caption{The DAG-construction recipe encoded in the pipeline's prompts.
The recipe is realized as a fixed sequence of five model calls rather than a literal node-by-node loop; the first call performs the full backward trace.}
\label{algo:dag_construction}
\begin{algorithmic}[1]
    \State Seed the graph with the sink event (discharge/death in this setting); the DAG describes the hospitalization as observed, and the intervention node $A$ is attached automatically after construction.
    \State Add the intervention targets, the events whose timing the intervention directly changes, and the sink-gating events (the direct parents of the sink).
    \State Trace backward from each sink-gating event, adding the intermediate events that determine its timing.
    \State Trace backward from each intervention target in the same fashion.
    \State Following the disjunctive cause criterion, retain the confounders: pre-exposure variables that cause an intervention target, a sink-gating event, or both, including a baseline patient-characteristics node.
    \State Record dependencies that link parallel care pathways, where an event is gated by an event on a different pathway; these become edges in the later steps.
    \State Finalize the DAG over the baseline, admission, intervention-target, intermediate, confounder, sink-gating, and sink nodes, with edges defined accordingly.
\end{algorithmic}
\end{algorithm}

\section{Prompt templates}

This section reproduces the two timing-query templates referenced in Section~\ref{sec:step_2}.
Curly-brace fields mark runtime substitutions, and the remaining pipeline prompts are available in the accompanying software package.

\subsection{Type III node timing query}
\label{sec:type_i_prompt}

This prompt elicits the revised timestamp of an intervention target (a Type~III node) under the proposed intervention.
At runtime, \texttt{intervention\_text} receives the intervention's label, description, and recommended change from its specification; \texttt{patient\_summary} a short summary of the hospitalization; \texttt{parent\_text} the target's parent events with their current timestamps; and \texttt{child\_text} the target event itself, without a timestamp.
GPT-5-mini returns a short rationale and a single revised timestamp, to which the capping rule of Appendix~\ref{sec:savings_appendix} is applied.

\begin{verbatim}
# TASK BACKGROUND
As a member of the hospital's Quality Improvement team, you have suggested
implementing the following operational intervention to reduce overall LOS in
the hospital:
"""
{intervention_text}
"""

Your task now is to look at a snippet of the patient's encounter that is
relevant to this intervention and determine the anticipated timestamp of the
event of interest, given recent events for this patient that are causally
relevant for determining the timestamp of this event in question.

# PATIENT INFO
Patient summary: {patient_summary}

What has happened recently for this patient that are relevant:
{parent_text}

The event that we need you to determine the timestamp for:
{child_text}

Consider:
1. Your planned intervention as described and how it would apply to this
   specific patient encounter
2. The parent node timestamps (the new timestamp cannot be earlier than any
   parent)

Output your reasoning and the new timestamp in JSON format:
{
  "reasoning": <your reasoning about how the intervention shifts this event>,
  "new_timestamp": "<YYYY-MM-DD HH:MM>"
}
\end{verbatim}

\subsection{Type II node timing query}
\label{sec:type_iii_prompt}

This prompt elicits the delay, in hours, between a downstream event (a Type~II node) and the latest of its parent events.
At runtime, \texttt{training\_csv} receives fifteen worked examples from other hospitalizations, matched on parent and child event type where possible and completed at random, each giving the patient summary, the parent and child event descriptors, and the observed delay.
The \texttt{patient\_summary}, \texttt{parent\_text}, and \texttt{child\_text} fields describe the query itself, and GPT-5 returns a step-by-step rationale and a single delay.

\begin{verbatim}
# TASK BACKGROUND
You are predicting the gap time between nodes in a hospital DAG that represents
the dependencies within an inpatient encounter, which will help us reason about
overall LOS.

# WHAT IS A DAG?
The DAG represents causal relationships between clinical events for a hospital
patient. DAG nodes are instantaneous clinical events and edges represent causal
dependencies.

# NODE TYPES
- start or end nodes of a process: indicates start/end of a clinical event
- state_achieved: Instantaneous event (start=end) indicating criteria are met
  (e.g. patient stable)
- action: Instantaneous decision/order (start=end) (e.g. PO ordered)
- start or end nodes of a waiting period: indicates start/end of a waiting
  period, potentially modifiable by intervention

# GAP TIME
Defined as the gap time between the child node and the maximum time among parent
nodes (i.e. time of latest event)

# TRAINING EXAMPLES
Here are examples from similar patient encounters:
{training_csv}

# YOUR TASK
Predict target_gap_hours for this edge:
{patient_summary}parent_text: {parent_text}
child_text: {child_text}

Think step by step about the clinical process, what the likely last parent event
is and what it represents. Use your clinical judgment on what is the most likely
length, along with looking at the provided examples of typical durations at this
hospital.
Output your reasoning and final prediction in hours.

Format your answer in JSON format:
{
  "reasoning": <your reasoning>,
  "prediction": <number in hours>
}
\end{verbatim}

\section{Simulation details} \label{sec:sim_details}

\paragraph{Estimation methods.}
We compare a set of estimation methods that differ along two axes: the causal structure they assume, either the patient's own DAG or a single shared DAG fit to the whole cohort, and how they determine the Type~II timings downstream of the intervention, either by exact event propagation or by a fitted model.
The methods built on the patient-specific DAG use the true latent structure recorded in the note, namely which reports the decision waits on and which report's timing depends on the other, an idealization that isolates the value of the patient-specific structure from the separate task of recovering the label from the note.
\noindent We describe each method in turn.
\begingroup
\renewcommand{\labelenumi}{(\roman{enumi})}
\begin{enumerate}
\item \textbf{Oracle}: Computes the counterfactual mechanically from the true data-generating structure by clipping the targeted delay at the cap.
This serves as the ground truth against which bias is measured and does not appear as a row in Table~\ref{tab:sim_results}.

\item \textbf{\texttt{DAG-extractor + Gantt-timing}}: Egg-computation using each patient's correct DAG, propagating the capped report time through the exact Gantt update that reuses the factual time differences (\texttt{Gantt-timing}, Section~\ref{sec:simulation}).
With the correct DAG this method recovers the truth, so it is the oracle configuration and an upper bound on what a perfect Type~II step can achieve.

\item \textbf{\texttt{DAG-extractor + OLS-timing}}, \textbf{\texttt{DAG-extractor + kNN-timing}}: Use the same correct DAG, but the Type~II timings are estimated rather than propagated exactly.
The OLS variant fits linear models for the downstream report and the decision, conditioned on which reports the patient's decision waits on and which report's timing depends on the other; the kNN variant estimates the counterfactual delay using a $k$-nearest-neighbor imputation from patients whose targeted delay already falls below the cap, mirroring the in-context retrieval the LLM Type~II step performs.
We select $k$ per intervention by leave-one-out cross validation on the observations already below the cap.
Per-patient error and ATE bias are insensitive to the choice of $k$ (Appendix~\ref{sec:knn_sweep}) because the kNN variant's residual error arises from estimating the mean of the below-cap observations rather than from the neighborhood size.

\item \textbf{\texttt{Single-DAG + Gantt-timing}}: The same exact Gantt update used by \texttt{DAG-extractor + Gantt-timing}, applied to a single shared DAG (every patient is assumed to wait on both reports with one common dependence direction).
Because the propagation of events has no estimation error, the two \texttt{Gantt-timing} methods share an identical, noiseless Type~II step and differ only in the assumed DAG, so any error results from the wrong DAG structure.
We report both orientations, CT$\to$MRI and MRI$\to$CT.

\item \textbf{\texttt{Single-DAG + OLS-timing}}: Fits a pooled g-computation model that assumes the same common DAG structure with a fitted linear Type~II step (ordinary least squares for the downstream report and the decision).
We report both dependence direction orientations, (CT$\to$MRI) and (MRI$\to$CT).

\end{enumerate}
\endgroup

We report each quantity averaged over 50 independent simulation replications, with its standard error across replications.

\section{Sensitivity of \texttt{DAG-extractor + kNN-timing} to the neighborhood size $k$}
\label{sec:knn_sweep}

The \texttt{kNN-timing} model imputes the counterfactual delay of each patient whose delay is capped from the $k$ nearest neighbors in baseline acuity, a baseline severity covariate in the simulation, among patients whose targeted delay already falls below the cap, so the neighborhood size $k$ is a tuning choice that could in principle drive its error.
We select $k$ per intervention by leave-one-out cross-validation among the observations with delays below the cap, predicting each donor's own delay from its $k$ nearest neighbors and minimizing held-out MAE.
This is the only model-selection signal available without the counterfactual, and it mirrors how a practitioner would set the number of in-context exemplars in the LLM Type~II step.

Figure~\ref{fig:knn_sweep} sweeps $k$ from a single neighbor to several hundred and shows that both the per-patient MAE and the ATE bias are flat to within $0.1$ hours across the entire grid, for both interventions, and far below the best-performing \texttt{Single-DAG} method throughout.
This is because every donor observation lies below the $12$-hour cap, so the kNN variant imputes a post-intervention delay below the cap regardless of the number of neighbors used.
Its residual error therefore arises from conditioning on observations below the cap, not from the choice of $k$.

\begin{figure}[t]
\centering
\includegraphics[width=\linewidth]{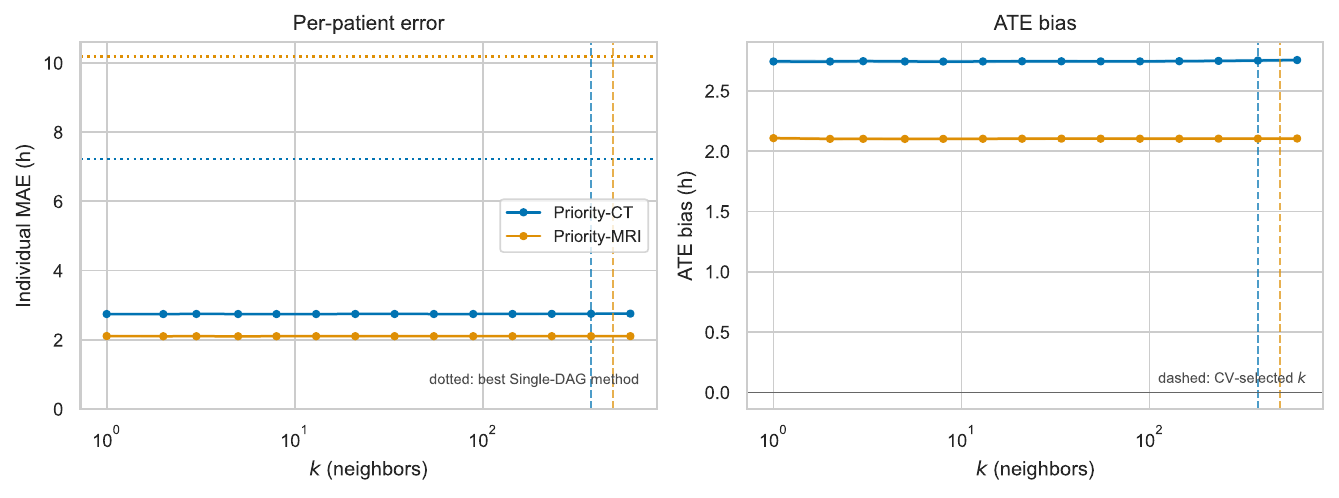}
\vspace{-1cm}
\caption{\texttt{DAG-extractor + kNN-timing} error as a function of the neighborhood size $k$, averaged over the 50 replications (bands are $\pm 1$ standard error). 
Left: per-patient MAE, with the dotted lines marking the lowest-MAE \texttt{Single-DAG} method for each intervention. 
Right: signed bias of the estimated average time saved 
The dashed vertical lines are the cross-validation-selected $k$. Both panels are flat in $k$, and the MAE stays far below the \texttt{Single-DAG} floor throughout, so the kNN variant's error reflects its estimation strategy rather than the choice of $k$.}
\label{fig:knn_sweep}
\end{figure}

\section{Definitions of the candidate interventions}
\label{sec:intervention_definitions}
The eleven candidates were drawn from a prior LLM-assisted analysis of the most frequent bottlenecks and delays at the hospital \citep{Vossler2026-ew}, and each was specified as an operational rule together with the inclusion and exclusion criteria that determine which hospitalizations are eligible.
Each candidate is operational in the sense that it changes when already-indicated work is completed, not what care the patient receives.
Each entry below reproduces the intervention's operational rule verbatim from its specification, followed by a summary of its eligibility scope; the complete specifications given to the LLM for eligibility determination and Type~III timing are provided in the accompanying software package.

\paragraph{Consult response escalation.}
\textbf{Rule:} For a prioritized ancillary consult (SLP, social work, PT, or OT) that gates discharge-relevant work, pending consult age is tracked on the flow navigator and escalated at 12 hours without a first note.
As such, order-to-first-evaluation should generally not exceed 12 hours, which encompasses acknowledgment of the order, the service reaching the patient, and the initial evaluation note.

\noindent The current median order-to-first-evaluation is about 1.5 days for the slowest of these services.
A hospitalization is eligible when such an ancillary consult is ordered and its evaluation gates functional clearance, the placement decision, or discharge planning.
Physician and medical-team consults (for example cardiology, neurology, or surgery) are not in scope, and a hospitalization is excluded when the patient could not be evaluated within the window for medical reasons, a clinical prerequisite had to occur first (such as extubation before a swallow evaluation), or the entire window falls on a weekend or holiday when the service does not staff evaluations.

\paragraph{Disposition-critical imaging priority.}
\textbf{Rule:} If a diagnostic imaging study (TTE, TEE, CT with contrast, MRI, or ultrasound) is needed for deciding patient disposition, the study is prioritized so that the study will be completed the same calendar day if the study is ordered before 3pm and will be completed by noon the following day if ordered after 3pm.
As such, order-to-completion for a flagged study should generally never exceed 24 hours, which encompasses the time of acknowledgment of the order, transportation of the patient, capturing the image, and the final report on the imaging study.

\noindent A hospitalization is eligible when one of these modalities is ordered and its result gates the disposition decision.
This is one of the two interventions used in the annotation study, and its full definition as shown to annotators is reproduced in Appendix~\ref{sec:annotation_protocol}.

\paragraph{Early discharge planning.}
\textbf{Rule:} Starting at the 1 day mark (24 hours), all patients will be evaluated for discharge planning.
Discharge planning will then proceed within the following 24 hours, and patients with complex discharge needs relating to housing, post-acute facility placement (SNF, acute rehab, LTACH), hospice, durable medical equipment, transportation, home health, IHSS, social-benefit applications (e.g., Medi-Cal), and/or substance use will meet with social work or case management within the 24 hour (i.e., 24--48 hours) time period.
Qualifying complex patients are: (1) unhoused patients or patients living in a shelter; (2) patients who came directly from a SNF or another facility/program; (3) patients who had home services; (4) patients who, based on the first 24 hours, will need to be set up with home services, a SNF, acute rehab, a long-term acute care facility (LTACH), or a hospice referral.
If the patient agrees, the social work/case manager will contact appropriate discharge resources within 24 hours after meeting with the patient.

\noindent A hospitalization is eligible when a complex-needs placement or discharge-services referral occurs and the need was identifiable within the first 24 hours of admission.
This is the second annotation-study intervention, and its full definition as shown to annotators is reproduced in Appendix~\ref{sec:annotation_protocol}.

\paragraph{Discharge before noon.}
\textbf{Rule:} At 4 PM each day, evening discharge rounds identify patients expected to discharge the following day; for those patients, the team enters the discharge order and prescriptions and completes readiness tasks that evening.
As such, on the discharge day itself, the discharge order should generally be in by 11 AM and the patient should leave before noon---the intervention moves the discharge earlier within the day, not to a different day.

\noindent The current median discharge time is near 4 PM.
A hospitalization is eligible when the patient left on a routine discharge and was documented by 4 PM the prior day as likely to discharge the next day.
It is excluded when the discharge still depended on a result or consult that could not be anticipated that evening, when new instability is documented overnight, or when the departure time was set by the receiving side rather than by discharge-task readiness.

\paragraph{Weekend discharge coordinator.}
\textbf{Rule:} Discharge coordinators staff Saturday and Sunday from 8 AM to 4 PM, rounding on patients flagged at Friday interdisciplinary rounds as likely to discharge within 48 hours and completing the coordination tasks that would otherwise wait for Monday.
As such, a patient who reaches medical readiness between Friday afternoon and Sunday, and whose remaining barriers are coordination tasks, should generally be discharged the same or next day rather than held to Monday.

\noindent A hospitalization is excluded when a clinical barrier remains, when the discharge requires a weekday-only service such as a pending consult or placement decision that coordinator coverage cannot resolve, or when the departure timing is set by the receiving facility.

\paragraph{Urgent IR procedure block.}
\textbf{Rule:} An IR triage physician reviews pending inpatient requests by 10 AM and directs urgent flow-related cases---active infection requiring drainage, an obstructed system with sepsis, vascular access for time-sensitive therapy---into 2 protected daily slots.
As such, request-to-completion for a qualifying urgent IR procedure should generally not exceed 24 hours.
The rule bounds the procedure, not the finalized report.

\noindent The current median request-to-completion is near one day, with the slowest tenth near six days.
Elective and outpatient cases are out of scope, and a hospitalization is excluded when the patient was not medically stable for the procedure, a clinical prerequisite had to occur first (such as anticoagulation reversal), or consent could not be obtained within the window.

\paragraph{Same-day echo completion.}
\textbf{Rule:} A discharge-prioritized TTE ordered before 11 AM is completed the same calendar day in a reserved slot; TEEs get add-on spots, and any TEE not completed within 24 hours of order is escalated to cardiology and anesthesia leadership---anesthesia scheduling delays are within the escalation's reach.
As such, order-to-completion for a discharge-prioritized TTE or TEE, including the report, should generally not exceed 24 hours.

\noindent A hospitalization is eligible when a TTE or TEE on the discharge critical path is ordered, for example for suspected endocarditis, new heart failure, or preoperative cardiac clearance.
It is excluded when the patient was not medically stable for the study, a clinical prerequisite had to occur first (such as NPO status for a TEE), or the 24-hour window falls entirely on a weekend or holiday when the reserved slots are unstaffed.

\paragraph{Antimicrobial stewardship.}
\textbf{Rule:} When a patient with sepsis on IV antibiotics meets oral conversion criteria---afebrile $>$24h, tolerating PO, clinically improving, and culture susceptibilities showing an oral option---an automated alert fires and a pre-authorized pharmacist protocol converts standard regimens without waiting for rounds.
As such, the time from meeting conversion criteria to the IV-to-PO conversion order should generally not exceed 12 hours.

\noindent The current median time from meeting criteria to conversion is about one day.
A hospitalization is excluded when the infection requires sustained IV therapy (for example endocarditis, bacteremia requiring prolonged IV, or a central-nervous-system infection), when the patient cannot tolerate oral intake, or when no oral-susceptible organism is available.

\paragraph{Weekend SNF intake.}
\textbf{Rule:} Partner SNFs accept weekend admissions with intake by 2 PM; case management completes referral packets by Thursday for patients anticipated to be SNF-ready over the weekend, partner facilities confirm bed availability by Friday noon, and transportation is pre-arranged.
As such, a patient with a SNF disposition identified by Thursday who reaches medical readiness between Friday and Sunday should generally transfer the same or next day rather than waiting until Monday.

\noindent A hospitalization is excluded when the SNF need arose too late to anticipate, when the patient requires acute rehabilitation or a long-term acute care hospital rather than a SNF, or when the patient needs services no partner facility provides.

\paragraph{Weekend PT/OT coverage.}
\textbf{Rule:} PT and OT therapists staff Saturday and Sunday from at least 8 AM to 2 PM, working from the Friday-rounds flag list and prioritizing patients who could discharge the same day if cleared.
As such, a patient pending only PT/OT clearance over the weekend should generally be evaluated by Sunday noon rather than waiting for Monday staffing.

\noindent A hospitalization is excluded when a clinical barrier beyond therapy clearance remains, when the patient's needs require inpatient-rehabilitation intensity rather than a clearance evaluation, or when the patient cannot participate in the evaluation at any staffed time.

\paragraph{Weekend diagnostic slots.}
\textbf{Rule:} Inpatient studies whose result gates discharge get reserved weekend capacity---at least 4 echo and 6 CT/MRI slots per weekend day, scheduled first-come from Friday noon, with anything unfinished by end of Saturday escalated to the on-call radiologist or cardiologist for Sunday morning.
As such, a discharge-critical study should generally be completed, including the report, within 24 hours of its weekend clock start---the order time for studies ordered after Friday noon, and Friday noon for studies ordered earlier and still pending then---rather than held to Monday.

\noindent A hospitalization is excluded when the patient was not medically stable for the study, when the study requires a weekday-only resource or specialized staffing the weekend lacks (such as anesthesia support for a TEE), or when a clinical prerequisite had to occur first.

\section{Computational cost of the LLM pipeline}
\label{sec:computational_cost}
Token usage was logged for every LLM call in the real-data pipeline; the figures below price those tokens at OpenAI list rates (GPT-5: \$1.25 per million input tokens and \$10 per million output tokens; GPT-5-mini: \$0.25 and \$2, respectively).
The full pipeline consumed 562 million input tokens and 90 million output tokens, for a total cost of approximately \$560 at standard rates, or roughly \$490 after prompt-caching discounts.
By stage, screening all 2{,}193 hospitalizations for eligibility against the eleven interventions cost \$29, constructing the 1{,}531 DAGs cost \$215 (about 14 cents per DAG), the Type III timing draws cost \$15, and the Type II ICL predictions (both the model evaluation and the counterfactual propagation) cost \$305.
The non-LLM timing baselines and the final averaging step incur no API cost.

\section{Annotation study protocol and metric definitions}
\label{sec:annotation_protocol}

\paragraph{Interface and review protocol.}
Annotators reviewed the LLM-generated DAGs through a web-based annotation interface that displays the events extracted for a hospitalization, the intervention description, and the intervention-relevant portion of the graph.
Each annotation task is one annotator's review of one intervention-hospitalization pair and proceeds in a fixed order.
First, the annotator confirms or rejects each event the LLM flagged as an intervention target and reviews the full list of extracted events to flag any targets the LLM missed.
Second, for each event reachable from the corrected intervention targets, the annotator approves or deletes each edge the LLM proposed into that event and reviews earlier events offered as candidate missing parents, adding any edges the LLM omitted.
Annotators corrected the LLM's graph rather than drawing their own from scratch, so all rates measure agreement under LLM-assisted review.
Figure~\ref{fig:annotation_interface} shows the interface, rendered on the synthetic test patient used during interface development.

\begin{figure}[htb]
\centering
\vspace{-0.4cm}
\includegraphics[width=\linewidth]{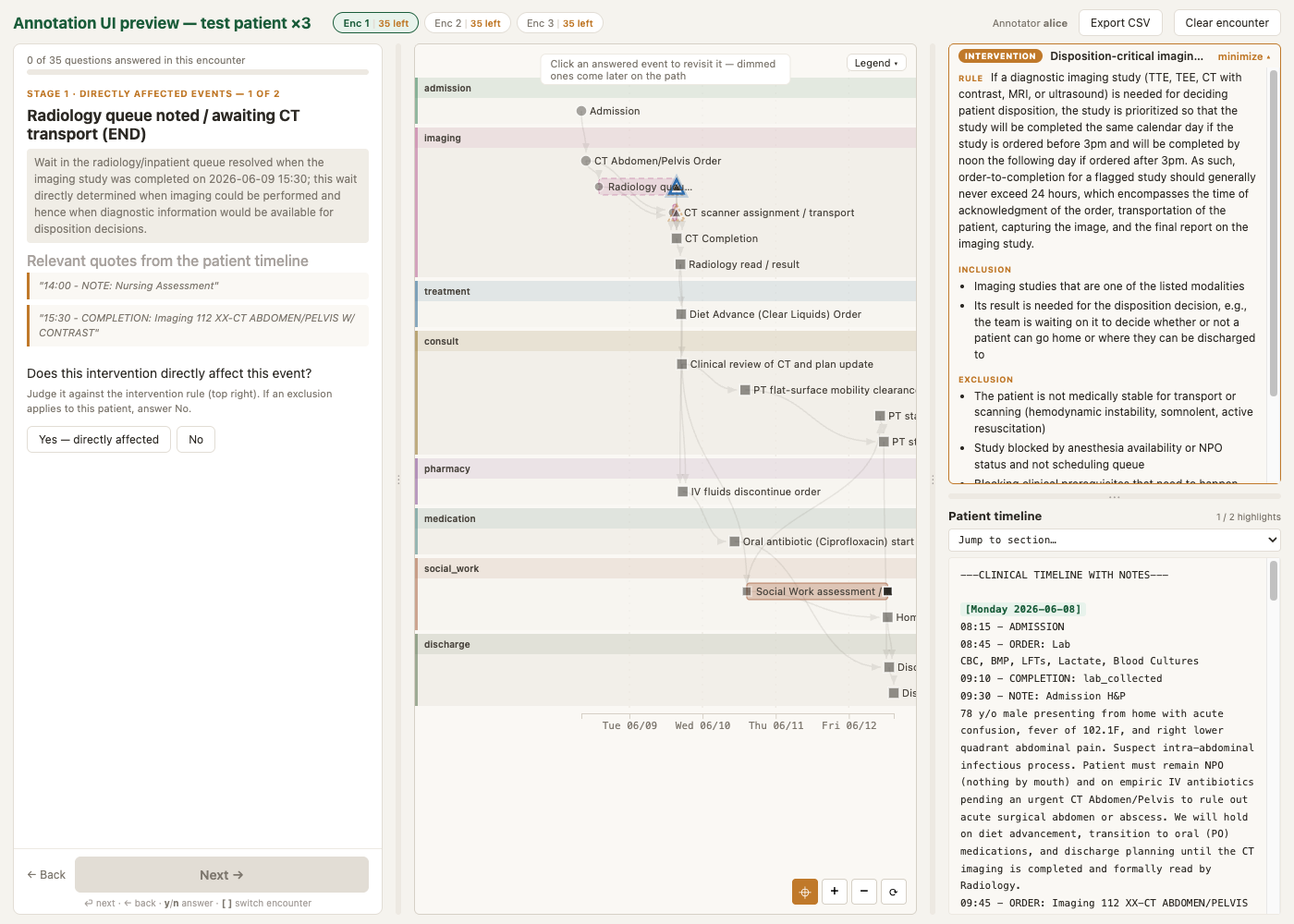}
\vspace{-0.9cm}
\caption{The annotation interface, rendered on the synthetic test patient used for interface development; no real patient data is shown.
The left panel poses one verdict question at a time together with the LLM's stated reasoning and its supporting timeline quotes.
The center panel shows the extracted DAG on a department-by-time canvas.
The right panel pins the intervention's rule and inclusion/exclusion criteria above the patient timeline.}
\label{fig:annotation_interface}
\end{figure}

The next two paragraphs reproduce the written guidance the annotators received.
Its vocabulary differs from the paper's in three places: the guide says ``encounter'' for hospitalization, ``QI intervention'' for candidate intervention, and ``directly affected event'' for intervention target (Type~III node).
References to ``the first tab'' point to the portions of the guide document reproduced under the conventions and intervention definitions below.

\paragraph{Annotator instructions.}
The remainder of this paragraph reproduces, verbatim, the written guidance given to the annotators, comprising the general annotation conventions and the step-by-step annotation guide for the interface.

The DAG indicates how event timings \textbf{causally} depend on one another.
There are two general tasks:
\begin{itemize}
  \item Would this event's timing be directly affected by the QI intervention? That is, would implementing the QI intervention directly change an event's timing?
  \item Is this event's timing directly dependent on a prior event's timing? That is, would a change in event A's timing directly cause event B's timing to change?
\end{itemize}

\noindent\textbf{Annotation conventions:}
\begin{itemize}
  \item A direct connection between events A and B means there is not a different event/node in the graph that would fully explain away the dependency between the timings of events A and B.
  \item Note that an intervention can still affect an event's timing if the event already happens very fast in the given patient, as the intervention may directly speed up this event's timing for a different patient.
  \begin{itemize}
    \item Answer ``no'' only when the mechanism doesn't apply or an exclusion applies
  \end{itemize}
  \item New events cannot be added to the DAG, you can only:
  \begin{itemize}
    \item change which events count as directly affected by the intervention
    \item add or remove connections between the displayed events
  \end{itemize}
\end{itemize}

Each encounter is shown as a graph of events extracted from the patient's chart, things like orders, results, consults, transfers, discharges, and any other events the LLM extracts from the notes.
An arrow between two events means the LLM believes the later event's timing directly depends on the earlier event's timing: a change in the earlier event's timing would directly cause the later event's timing to change.
In addition to the graph you have access to the patient timeline with events and the notes that the LLM saw when constructing the graph.

Your job is to check the model's constructed graph for the intervention shown in the first tab of this doc and in the annotation interface.
The tool goes through questions one at a time and you answer yes or no.
Your reasoning should be based on what the chart documents and the question is always ``would this intervention, as written, directly affect this event for this patient?''

\noindent\textbf{Task 1: which events would the intervention affect directly?}

\noindent The interface first asks you to review nodes in the DAG to determine if the LLM correctly marked events that are directly affected by the intervention.
\begin{itemize}
  \item ``Directly affected'' means the intervention itself acts on that event's timing. Events that would move only because an earlier event moved don't count; the model handles those downstream changes.
  \item An event that was already fast still gets a yes if the mechanism applies to it (see the conventions on the first tab).
  \item When you answer no, add a few words on why: which exclusion applies, or why the mechanism doesn't act on this event.
\end{itemize}
After the LLM-flagged events, the interface asks whether the intervention would directly affect any other events in the graph.
If there are any directly-affected events that were missed, you can add them now.

\noindent\textbf{Task 2: which event timings does this event timing depend on?}

\noindent Next you review the edges of the DAG to determine if event timing dependencies are correctly determined by the LLM.
\begin{itemize}
  \item Keep an edge to the parent event if the event truly had to wait for it. Had the parent happened earlier, this event could have happened sooner too
  \item Remove the edge if there is no causal dependence between them
  \item Add a missing edge if an earlier event in the graph blocked this one and isn't listed.
  \item New events cannot be added. You can only add or remove connections between the events shown in the DAG
  \item Note that the final (sink) node in the DAG is always ``Discharge/Death'', which indicates the end of the hospital stay where the patient was either discharged or died. Which of the two should be evident from the patient timeline and LLM extractions.
\end{itemize}

\noindent\textbf{When something looks off}
\begin{itemize}
  \item If an event looks incorrect (mislabeled, incorrect time, not supported by the notes), describe the problem in the ``no'' note
  \item If the intervention doesn't apply to this patient at all, answer no for each flagged event and explain once in the first note.
\end{itemize}

\paragraph{Intervention definitions shown to annotators.}
The two interventions used in the annotation study were presented to annotators with the following definitions, reproduced here verbatim.

\medskip
\noindent\textbf{Disposition-critical imaging priority}

\noindent\textbf{Rule:} If a diagnostic imaging study (TTE, TEE, CT with contrast, MRI, or ultrasound) is needed for deciding patient disposition, the study is prioritized so that the study will be completed the same calendar day if the study is ordered before 3pm and will be completed by noon the following day if ordered after 3pm.
As such, order-to-completion for a flagged study should generally never exceed 24 hours, which encompasses the time of acknowledgment of the order, transportation of the patient, capturing the image, and the final report on the imaging study.

\noindent\textbf{Inclusion Criteria:}
\begin{itemize}
  \item Imaging studies that are one of the listed modalities
  \item Its result is needed for the disposition decision, e.g., the team is waiting on it for critical diagnostic work-ups, to decide whether or not a patient can be discharged, or where they can be discharged to.
\end{itemize}

\noindent\textbf{Exclusion Criteria:}
\begin{itemize}
  \item The patient is not medically stable for transport or scanning
  \begin{itemize}
    \item Hemodynamic instability, somnolence, active resuscitation
  \end{itemize}
  \item Study blocked by anesthesia availability or study prerequisites (e.g. NPO status) and not scheduling queue
  \item Blocking clinical prerequisites that need to happen first, such as completion of the treatments, required follow-up intervals, or medical improvements.
  \item Documented patient factors that prevent the study from happening sooner (refusal, agitation, non-cooperation, patient in another procedure, documented allergies to contrasts, etc.)
\end{itemize}

\medskip
\noindent\textbf{Early Discharge planning}

\noindent\textbf{Rule:} Starting at the 1 day mark (24 hours), all patients will be evaluated for discharge planning.
Discharge planning will then proceed within the following 24 hours and patients with complex discharge needs relating to housing, post-acute facility placement (SNF, acute rehab, LTACH), hospice, durable medical equipment, transportation, home health, IHSS, social-benefit applications (e.g., Medi-Cal), and/or substance use will meet with social work or case management within the 24 hour (i.e., 24--48 hours) time period.
Qualifying complex patients are:
\begin{itemize}
  \item Unhoused/live in a shelter
  \item Patients who came directly from a SNF/facilities/programs
  \item Patients who had home services
  \item Patients who will need to be set up with home services or a SNF or acute rehab or Long-term acute care facility (LTACH) or hospice referral based on the first 24 hours
\end{itemize}
If the patient agrees, the social work/case manager will contact appropriate discharge resources within 24 hours after meeting with the patient.

\noindent\textbf{Inclusion Criteria:}
\begin{itemize}
  \item A post-acute placement referral happens during the encounter.
  \item The placement need was identifiable within 24 hours of admission
\end{itemize}

\noindent\textbf{Exclusion Criteria:}
\begin{itemize}
  \item The placement need came up later in the encounter (e.g., a new stroke, deterioration, or new functional decline) so identification does not occur within the first 24 hours
  \item The patient has been discharged or dies in the hospital within 24 hours
  \item Meeting with the social work/case manager within the 24--48 hour window is not possible because of medical reasons or legal reasons, e.g., patient is not responding and no family members available, guardianship/conservatorship proceedings
  \item Patients who cannot be medically assessed for disposition placement (e.g., needs PT/OT/SLP)
  \item Meeting with the social work/case manager within the 24--48 hour window is not possible because it's a holiday/weekend
  \item Discharge needs related to medications (e.g., discharge medications and medication reconciliation) are not accelerated through this intervention.
\end{itemize}

\paragraph{Human-versus-LLM metrics.}
Intervention-target specificity is the number of events the annotator judged non-targets that the LLM also left unflagged, divided by the number of events the annotator judged non-targets.
Intervention-target recall is the number of annotator-endorsed targets the LLM flagged, divided by the number of annotator-endorsed targets.
Both metrics are computed over every event extracted for the hospitalization, not only the events the LLM flagged.
Edge precision is the number of LLM-proposed edges the annotator kept, divided by the number of LLM-proposed edges the annotator explicitly reviewed.
Edge recall is the number of annotator-endorsed edges present in the LLM's graph, divided by the number of annotator-endorsed edges.
As a broader secondary denominator, we also compute an edge specificity that counts every unselected candidate parent as a negative judgment; this yields 97.9\% pooled across interventions.
Point estimates weight each intervention-hospitalization pair by the inverse of its annotation multiplicity, so a double-annotated pair contributes total weight one.
Confidence intervals are computed by a cluster bootstrap over hospital stays, where the same hospitalization appearing under both interventions resamples jointly, and are conditional on the five annotators.
Edge additions concentrate in individual annotators, so edge recall is sensitive to the panel; recomputing it with any one annotator removed moves it between roughly 79\% and 92\%.

\paragraph{Estimate-agreement setup.}
For the comparison of estimated time saved on annotator-corrected versus LLM-generated graphs (final row of Table~\ref{tab:llm_is_ok}), we re-ran egg-computation on each graph while holding the timing models fixed.
The LLM's Type III timing predictions are applied to each graph's own intervention targets and capped at the observed times, since operational interventions remove delays rather than create them.
The Type II step is deterministic, with each downstream event occurring as soon as its last parent does.
Under this setup, any difference in the estimated time saved traces to the graphs alone rather than to timing-model noise.

\paragraph{Inter-annotator metrics.}
Inter-annotator agreement is computed on the eight double-annotated pairs.
Edge units are the edges both annotators explicitly reviewed, coded keep versus delete; one-sided additions and differences in reviewed scope are not measured.
Node units are the intervention-target status of every event in the hospitalization.
We report raw percent agreement, Krippendorff's $\alpha$, Gwet's AC$_1$, and within-class agreement \citep{Gwet2008-cn, Hayes2007-pu}.
For a rating class $c$ with $n_c$ concordant pairs and $d_c$ discordant pairs involving a $c$ rating, within-class agreement is $2n_c / (2n_c + d_c)$, the share of $c$ ratings matched by the other annotator \citep{Feinstein1990-pw, Cicchetti1990-dz}.
We use the nominal versions of both chance-corrected coefficients because annotator pairings varied across tasks, so no stable rater-one-versus-rater-two assignment exists for Cohen's $\kappa$.
Under the keep-skewed rating distribution, $\alpha$ and AC$_1$ diverge because their chance models differ, and the within-class rates are what localize the disagreement to deletions.
All inter-annotator statistics are descriptive rather than inferential, as they pool eight pairs covering seven distinct hospitalizations and condition on this annotator panel.

\section{Type II timing evaluation details}
\label{sec:type_iii_appendix}

\paragraph{Protocol.}
Each method predicts the delay, in hours, between a downstream event (Type II node) and the last of its parent events to complete, given the parent events with their timestamps, the event's own description without its timestamp, and a patient summary.
Under leave-one-hospitalization-out cross-validation, every method trains on the delays of all other hospitalizations, pooled across the five diagnosis groups within the same intervention, and predicts the held-out hospitalization's delays.
The ICL model receives fifteen worked examples per query, drawn first from training delays whose parent and child event types match the query and topped up at random, and returns a numeric delay with a short rationale (the prompt template is given in Appendix~\ref{sec:type_iii_prompt}); the same GPT-5 configuration is used for every intervention.
The zero baseline predicts no delay; pair-type median predicts the training median among delays with the same parent and child event types; ridge and random forest use event-type indicators and temporal features (hour of day, day of week, time since admission); embedding + ridge adds sentence embeddings of the same text the ICL model reads, so the strongest regression baseline has the same information access.

\paragraph{Cleaning descriptions to remove timing-revealing language.}
The event descriptions and patient summaries the models read are scrubbed of timing-revealing language (clock times, calendar dates, weekday names, and spelled-out durations) so the target delay cannot leak into the inputs; parent timestamps are retained because they are legitimate inputs.
An automated audit scans every model-visible field for the same patterns and reports zero flagged rows across all 4{,}790 evaluated nodes.
\label{sec:type_iii_leakage}

\paragraph{Per-intervention accuracy.}
Table~\ref{tab:type_iii_by_intervention} breaks the pooled comparison of Table~\ref{tab:type_iii_results} out by intervention.
ICL attains the lowest MAE for nine of the eleven interventions, with every paired interval clear of zero for eight.
The exceptions concentrate where delays are near-immediate: discharge before noon (median delay 1.3 hours) is the one intervention where the zero baseline is numerically better, and for early discharge planning and the weekend discharge coordinator the closest baseline is statistically indistinguishable from ICL.
\label{sec:type_iii_by_intervention}

\begin{table}[htb]
\caption{Per-intervention Type II timing accuracy (hours).
$\Delta$MAE is the paired difference in MAE between a baseline and ICL (baseline minus ICL, so positive values favor ICL), shown only for each intervention's closest baseline, i.e., the competitor with the smallest $\Delta$MAE.
The final column counts how many of the five baseline comparisons have 95\% intervals lying entirely above zero, so 5/5 means ICL beats every baseline with an interval clear of zero.}
\label{tab:type_iii_by_intervention}
\resizebox{\linewidth}{!}{%

\begin{tabular}{lrrccccc}
    \toprule
    Intervention & Nodes & Hosp. & ICL MAE & ICL Med.\ AE & Closest baseline & $\Delta$MAE [95\% CI] & CIs $>$ 0 \\
    \midrule
    Discharge before noon & 206 & 141 & 6.2 & 1.3 & Zero gap & -0.7 [-2.0, 0.5] & 2/5 \\
    Weekend diagnostic slots & 453 & 131 & 11.8 & 2.5 & Pair-type median & +1.3 [0.6, 2.0] & 5/5 \\
    Weekend discharge coordinator & 253 & 148 & 11.9 & 1.6 & Zero gap & +1.1 [-0.8, 3.3] & 3/5 \\
    Antimicrobial stewardship & 238 & 109 & 18.6 & 2.6 & Zero gap & +2.8 [1.2, 4.7] & 5/5 \\
    Weekend PT/OT coverage & 515 & 145 & 21.5 & 2.0 & Pair-type median & +2.1 [1.2, 3.1] & 5/5 \\
    Same-day echo completion & 514 & 147 & 21.6 & 2.9 & Pair-type median & +2.8 [1.6, 4.3] & 5/5 \\
    Weekend SNF intake & 237 & 118 & 22.5 & 2.0 & Zero gap & +2.7 [0.7, 4.7] & 5/5 \\
    Disposition-critical imaging priority & 711 & 149 & 23.3 & 3.9 & Pair-type median & +3.3 [2.0, 4.5] & 5/5 \\
    Early discharge planning & 382 & 150 & 24.4 & 3.3 & Pair-type median & -0.2 [-2.2, 1.6] & 3/5 \\
    Consult response escalation & 582 & 150 & 34.5 & 4.0 & Pair-type median & +5.4 [3.4, 7.3] & 5/5 \\
    Urgent IR procedure block & 699 & 120 & 40.0 & 14.9 & Pair-type median & +6.6 [4.8, 8.6] & 5/5 \\
    \bottomrule
\end{tabular}
}
\end{table}

\paragraph{Accuracy by delay stratum.}
Table~\ref{tab:type_iii_by_stratum} reports MAE and median AE within strata of the observed delay.
The zero baseline is nearly exact on the 40\% of delays under one hour and remains strong to six hours, while the regression models hedge toward mid-range predictions that lose badly there.
ICL concedes little on the short strata, is clearly best between six hours and one day, and trails only the regression models' small edge in the extreme tail.
Within-stratum MAE is dominated by a small share of large errors, so the median gives the typical error; for example, ICL's 4.2-hour MAE on the shortest stratum reflects a median error of 1.5 hours alongside the 12\% of nodes missed by more than 10 hours.
\label{sec:type_iii_strata}

\begin{table}[htb]
\caption{MAE (median AE in parentheses, both in hours) by observed-delay stratum, pooled across interventions.}
\label{tab:type_iii_by_stratum}
\resizebox{\linewidth}{!}{%

\begin{tabular}{lrcccccc}
    \toprule
    Delay stratum & Nodes & Zero gap & Pair-type median & Ridge & Random forest & Embedding + ridge & ICL (GPT-5) \\
    \midrule
    $\leq$1h & 1904 & 0.1 (0.0) & 4.6 (1.8) & 24.0 (20.6) & 23.8 (13.9) & 21.5 (13.6) & 4.2 (1.5) \\
    1--6h & 914 & 2.9 (2.6) & 4.7 (2.0) & 22.3 (19.0) & 23.6 (12.9) & 21.1 (11.7) & 4.3 (1.3) \\
    6--24h & 715 & 16.0 (17.7) & 11.9 (11.3) & 18.5 (13.5) & 26.1 (15.7) & 23.8 (17.9) & 10.0 (5.9) \\
    $>$24h & 1257 & 96.5 (62.6) & 86.7 (50.8) & 67.6 (34.0) & 71.3 (37.3) & 67.3 (35.9) & 76.8 (44.5) \\
    \bottomrule
\end{tabular}
}
\end{table}

\paragraph{Structural validity sensitivity.}
About 10\% of the downstream node timings (509 of 4{,}790) come from DAGs that fail an automated structural check (temporal inconsistencies, cycles, or events with no path to discharge); these timings remain in the primary analysis.
Table~\ref{tab:type_iii_validity} recomputes the pooled comparison excluding them: every MAE moves by less than half an hour and every interval conclusion is unchanged.
\label{sec:type_iii_validity}

\begin{table}[htb]
\caption{Structural-validity sensitivity for the pooled Type II timing comparison.}
\label{tab:type_iii_validity}
\resizebox{\linewidth}{!}{%

\begin{tabular}{llccc}
    \toprule
    Population & Method & MAE & Med.\ AE & $\Delta$MAE vs.\ ICL [95\% CI] \\
    \midrule
    All DAGs (4790 nodes, 736 hosp.) & Zero gap & 28.30 & 2.74 & +4.18 [3.57, 4.77] \\
     & Pair-type median & 27.23 & 4.63 & +3.11 [2.57, 3.67] \\
     & Ridge & 34.30 & 20.92 & +10.18 [9.37, 10.99] \\
     & Random forest & 36.58 & 17.73 & +12.46 [11.18, 13.90] \\
     & Embedding + ridge & 33.79 & 18.90 & +9.67 [8.57, 10.70] \\
     & ICL (GPT-5) & 24.12 & 3.00 & --- \\
    \addlinespace
    Structurally valid only (4281 nodes, 701 hosp.) & Zero gap & 28.57 & 2.85 & +4.23 [3.60, 4.89] \\
     & Pair-type median & 27.27 & 4.63 & +2.94 [2.38, 3.50] \\
     & Ridge & 34.22 & 20.55 & +9.89 [9.04, 10.75] \\
     & Random forest & 36.59 & 17.73 & +12.25 [10.76, 13.90] \\
     & Embedding + ridge & 33.87 & 18.51 & +9.53 [8.43, 10.70] \\
     & ICL (GPT-5) & 24.34 & 3.00 & --- \\
    \addlinespace
    \bottomrule
\end{tabular}
}
\end{table}

\section{Additional Results}

\begin{table}
    \centering
    \caption{
    Inter-annotator agreement on the double-annotated hospitalizations.
    Edge units are the edges both annotators explicitly reviewed; node units are the intervention-target status of every event in the hospitalization.}
    \label{tab:interannotator}
    \resizebox{\linewidth}{!}{%
    \begin{tabular}{lccccc}
    \toprule
     & $n$ & Raw agreement & Krippendorff's $\alpha$ & Gwet's AC$_1$ & Within-class agreement \\
    \midrule
    Edges (keep vs.\ delete) & 45 & 82.2\% & 0.11 & 0.78 & keep 90.0\%, delete 20.0\% \\
    Nodes (intervention target vs.\ not) & 116 & 97.4\% & 0.89 & 0.97 & non-target 98.5\%, target 90.3\% \\
    \bottomrule
    \end{tabular}%
    }
\end{table}

\paragraph{Inter-rater evaluation results.}
Table~\ref{tab:interannotator} reports agreement between annotators on the double-annotated intervention-hospitalization pairs.
Seven of the eight pairs produced identical sets of intervention targets, and raw agreement on the keep-versus-delete decision for commonly reviewed edges was similarly high.
Because roughly 89\% of the individual edge ratings are in the ``keep'' group, the two chance-corrected coefficients diverge under this skew (Krippendorff's $\alpha = 0.11$, Gwet's AC$_1$ = 0.78) \citep{Feinstein1990-pw, Cicchetti1990-dz}, so we follow \citep{Cicchetti1990-dz} and report agreement within each rating class, i.e., how often a rating of a given class was matched by the other annotator.
Annotators matched 90\% of keep ratings but only 20\% of deletions, and seven of the eight discordant edges concern the same judgment of whether an event that precedes discharge actually gates the discharge decision or merely happens before it.
Disagreements over intervention targets likewise reflect differing interpretations of an eligibility criterion; annotators applied different thresholds for when a patient's discharge needs are identifiable within the first day of the stay, which determines whether early discharge planning applies to the patient at all.
Both are due to differences in interpretation of the intervention specification rather than disagreements about the DAG generated by the LLM.

\begin{table}[htb]
\caption{Estimated average time saved per eligible hospitalization (hours), by intervention, from egg-computation with the ICL timing model, for only the eligible interventions.
Intervals are 95\% bootstrap intervals over hospitalizations.
Reference variants and per-diagnosis breakdowns are given in Appendix~\ref{sec:savings_appendix}.}
\label{tab:intervention_savings}
\centering

\begin{tabular}{lrcc}
    \toprule
    Intervention & $n$ & Mean savings (h) & Median (h) \\
    \midrule
    Consult response escalation & 150 & 26.6 [17.0, 37.8] & 0.9 \\
    Disposition-critical imaging priority & 150 & 19.6 [12.8, 27.8] & 0.0 \\
    Early discharge planning & 150 & 15.7 [7.2, 25.6] & 0.0 \\
    Discharge before noon & 150 & 14.6 [9.2, 20.8] & 1.0 \\
    Weekend discharge coordinator & 150 & 19.6 [10.9, 30.9] & 0.6 \\
    Urgent IR procedure block & 120 & 49.9 [27.0, 78.7] & 0.0 \\
    Same-day echo completion & 150 & 22.9 [11.7, 37.0] & 0.0 \\
    Antimicrobial stewardship & 110 & 10.6 [4.2, 17.9] & 0.0 \\
    Weekend SNF intake & 123 & 21.3 [12.7, 32.0] & 1.3 \\
    Weekend PT/OT coverage & 145 & 11.2 [5.4, 18.9] & 0.0 \\
    Weekend diagnostic slots & 133 & 13.8 [6.7, 21.7] & 0.0 \\
    \bottomrule
\end{tabular}
 \end{table}

\section{Clinician review of the estimated rankings}
\label{sec:clinician_review}

\paragraph{Reviewer panel.}
Four clinicians from the hospital's QI team reviewed the estimates, two of them authors of this work (AH, LZ) and two of them not.
Reviewers worked independently and did not see one another's responses.

\paragraph{Interface and review protocol.}
Reviewers completed three tasks asynchronously through a self-contained web page.
The page put the tasks in a fixed order and withheld our estimates until the first was finished.

\paragraph{Task 1: ranking the shortlist.}
The three interventions our estimates ranked highest by expected impact were shown in an order that carried no information about our estimated ranking.
Each was shown with a one-sentence description of the operational rule and the current delay it targets.
The remaining eight interventions were also included for reference.
Reviewers were asked to rank the three, with the instruction ``Before seeing the order our estimates put them in, rank the three yourself: 1st = most total inpatient time saved.''
Only after all three were ranked did the page reveal the full ranking of eleven candidates, reporting for each the eligibility rate, the estimated time saved per eligible hospitalization, and the expected impact, together with a slopegraph pairing the eligibility and expected-impact rankings and a scatter of the two quantities whose product is the expected impact.
Reviewers were then asked which of the estimated rankings surprised them and which matched their experience.
Table~\ref{tab:clinician_orderings} gives the orderings from the reviewers.

\paragraph{Task 2: reviewing example hospitalizations.}
Each reviewer was assigned three interventions: the top two interventions, consult response escalation and disposition-critical imaging priority, and one random intervention.
For each assigned intervention, the page showed the eligibility rate, the estimated time saved per eligible hospitalization, the expected impact, the share of eligible hospitalizations with no estimated benefit, and the median saving among those that benefit (see Figure~\ref{fig:example_task_2}).
It also showed the eligibility rate within each diagnosis group and a histogram of estimated savings across eligible hospitalizations.
Below these summaries the page gave an LLM-written account of the mechanism, separating the hospitalizations in which the intervention shortens the stay from those in which discharge remains gated by some other event, together with the number of hospitalizations of each kind.
Four example hospitalizations followed, each rendered as the observed and counterfactual timings of every event in its DAG, with the intervention targets marked and each event's incoming edges shown on hover.
Reviewers were asked two questions of each assigned intervention:
\begin{itemize}
  \item ``Are the generated analyses clinically plausible and helpful in interpreting the likely impact of this QI intervention? Please provide detailed feedback since this will help us improve the interface and understand how well the current pipeline is working.''
  \item ``Based on the generated analyses, would you consider modifying the intervention further to improve its efficacy? If so, please explain how you would change it and what in the generated analyses has prompted you to consider this modification.''
\end{itemize}

For each intervention we drew two hospitalizations at random from those whose stay the estimate shortens and two at random from those whose estimated saving is within one hour of zero.
The examples are therefore balanced between the two outcomes by construction and do not reflect the share of eligible hospitalizations that benefit, which the accompanying histogram reports separately.

\paragraph{Task 3: general comment.}
Reviewers were asked two closing questions: ``Do you have general feedback on the QI intervention analysis pipeline? How can the analysis be improved further?'' and ``Are there other data sources that you think would be important to add to our pipeline (we currently analyze only orders, notes, results, and their timestamps)?''

\begin{table}[htb]
\centering
\caption{Orderings of the three highest-ranked interventions returned in Task 1, each produced before the reviewer saw any estimate, against the ordering by estimated expected impact.}
\label{tab:clinician_orderings}
\begin{tabular}{llll}
\toprule
 & First & Second & Third \\
\midrule
R1 & Consult response escalation & Disposition-critical imaging & Early discharge planning \\
R2 & Consult response escalation & Disposition-critical imaging & Early discharge planning \\
R3 & Early discharge planning & Disposition-critical imaging & Consult response escalation \\
R4 & Early discharge planning & Consult response escalation & Disposition-critical imaging \\
\midrule
Estimated & Consult response escalation & Disposition-critical imaging & Early discharge planning \\
Eligibility & Disposition-critical imaging & Consult response escalation & Early discharge planning \\
\bottomrule
\end{tabular}
\end{table}

\section{Time-saved estimation details}
\label{sec:savings_appendix}

\paragraph{Results.}

\begin{table}[htb]
\caption{
Prevalence-based versus causal prioritization of the eleven candidate interventions.
Eligibility rate is the share of hospitalizations that meet an intervention's inclusion criteria.
Expected time saved is the causal effect estimate from egg-computation.
Interventions are ordered by expected impact.
Brackets give 95\% percentile bootstrap intervals, resampling the screened cohort within each diagnosis group and the eligible hospitalizations within each intervention.
}
\label{tab:prevalence_vs_savings}
\centering
\begin{tabular}{lrrrr}
    \toprule
    & \multicolumn{2}{c}{Eligibility rate} & \multicolumn{2}{c}{Expected impact} \\
    \cmidrule(lr){2-3} \cmidrule(lr){4-5}
    Intervention & \% [95\% CI] & rank & h/screened [95\% CI] & rank \\
    \midrule
    Consult response escalation & 71 [69, 73] & 2 & 18.9 [11.4, 26.9] & 1 \\
    Disposition-critical imaging priority & 80 [78, 81] & 1 & 15.6 [10.2, 22.1] & 2 \\
    Early discharge planning & 69 [67, 70] & 3 & 10.8 [5.0, 17.5] & 3 \\
    Discharge before noon & 42 [40, 44] & 4 & 6.2 [3.9, 8.6] & 4 \\
    Weekend discharge coordinator & 26 [24, 28] & 5 & 5.0 [2.7, 8.0] & 5 \\
    Urgent IR procedure block & 10 [8, 11] & 8 & 4.8 [2.5, 7.4] & 6 \\
    Same-day echo completion & 17 [15, 18] & 7 & 3.9 [2.0, 6.5] & 7 \\
    Antimicrobial stewardship & 21 [20, 22] & 6 & 2.2 [0.9, 3.7] & 8 \\
    Weekend SNF intake & 8 [7, 9] & 10 & 1.7 [1.0, 2.6] & 9 \\
    Weekend PT/OT coverage & 9 [8, 10] & 9 & 1.0 [0.5, 1.7] & 10 \\
    Weekend diagnostic slots & 7 [6, 8] & 11 & 0.9 [0.4, 1.6] & 11 \\
    \bottomrule
\end{tabular}
 \end{table}

\paragraph{Rank stability.}
Table~\ref{tab:rank_stability} quantifies which ranking statements survive resampling.
Each of 2{,}000 bootstrap resamples redraws the screened cohort within each diagnosis group (jointly across interventions, preserving the correlation of eligibility within a hospitalization) and the eligible hospitalizations within each intervention, then recomputes both rankings.
The top of the expected-impact ranking is stable: consult response escalation is first in 72\% of resamples and disposition-critical imaging priority in 24\%, and the same three interventions occupy the top three in 89\% of resamples, while the bottom three are drawn almost entirely from the weekend interventions.
\label{sec:rank_stability}

\begin{table}[htb]
\caption{Rank stability across 2{,}000 bootstrap resamples: the probability each intervention lands in the top three of each ranking, is first by expected impact, or is in the bottom three by expected impact.}
\label{tab:rank_stability}
\resizebox{\linewidth}{!}{%
\begin{tabular}{lcccc}
    \toprule
    & \multicolumn{2}{c}{Top 3 probability} & & \\
    \cmidrule(lr){2-3}
    Intervention & Prevalence & Expected impact & $P(\text{rank 1, expected impact})$ & $P(\text{bottom 3, expected impact})$ \\
    \midrule
    Consult response escalation & 1.00 & 1.00 & 0.72 & 0.00 \\
    Disposition-critical imaging priority & 1.00 & 1.00 & 0.24 & 0.00 \\
    Early discharge planning & 1.00 & 0.89 & 0.03 & 0.00 \\
    Discharge before noon & 0.00 & 0.07 & 0.00 & 0.00 \\
    Weekend discharge coordinator & 0.00 & 0.03 & 0.00 & 0.00 \\
    Urgent IR procedure block & 0.00 & 0.02 & 0.00 & 0.00 \\
    Same-day echo completion & 0.00 & 0.00 & 0.00 & 0.01 \\
    Antimicrobial stewardship & 0.00 & 0.00 & 0.00 & 0.28 \\
    Weekend SNF intake & 0.00 & 0.00 & 0.00 & 0.72 \\
    Weekend PT/OT coverage & 0.00 & 0.00 & 0.00 & 0.99 \\
    Weekend diagnostic slots & 0.00 & 0.00 & 0.00 & 1.00 \\
    \bottomrule
\end{tabular}
}
\end{table}

\begin{table}[htb]
\caption{Primary estimate of average time saved (hours) by intervention and diagnosis group; cell entries are mean (number of eligible hospitalizations in the analysis subsample).}
\label{tab:savings_by_dx}
\resizebox{\linewidth}{!}{%

\begin{tabular}{lrrrrrr}
    \toprule
    Intervention & A41 & I63 & L03 & S06 & T40 & Pooled \\
    \midrule
    Consult response escalation & 29.5 (30) & 52.1 (30) & 12.4 (30) & 21.1 (30) & 17.9 (30) & 26.6 (150) \\
    Disposition-critical imaging priority & 28.0 (30) & 6.8 (30) & 24.5 (30) & 26.7 (30) & 11.9 (30) & 19.6 (150) \\
    Early discharge planning & 19.6 (30) & 37.7 (30) & 9.1 (30) & 3.0 (30) & 9.4 (30) & 15.7 (150) \\
    Discharge before noon & 14.8 (30) & 24.3 (30) & 7.5 (30) & 5.8 (30) & 20.5 (30) & 14.6 (150) \\
    Weekend discharge coordinator & 28.1 (30) & 12.0 (30) & 9.4 (30) & 41.1 (30) & 7.5 (30) & 19.6 (150) \\
    Urgent IR procedure block & 61.9 (30) & 82.0 (30) & 80.0 (7) & 22.9 (30) & 18.4 (23) & 49.9 (120) \\
    Same-day echo completion & 29.7 (30) & 22.0 (30) & 11.5 (30) & 46.9 (30) & 4.3 (30) & 22.9 (150) \\
    Antimicrobial stewardship & 15.1 (30) & 5.7 (12) & 4.5 (30) & 9.7 (8) & 14.4 (30) & 10.6 (110) \\
    Weekend SNF intake & 15.3 (30) & 30.2 (30) & 5.6 (11) & 17.8 (30) & 29.8 (22) & 21.3 (123) \\
    Weekend PT/OT coverage & 10.0 (30) & 3.3 (30) & 7.0 (25) & 21.3 (30) & 13.8 (30) & 11.2 (145) \\
    Weekend diagnostic slots & 28.2 (30) & 16.8 (30) & 14.8 (30) & -1.1 (16) & 2.1 (27) & 13.8 (133) \\
    \bottomrule
\end{tabular}
}
\end{table}

\section*{Acknowledgments}
The PROSPECT lab (PV, JO, LZ, JF) thanks Zuckerberg Priscilla Chan quality improvement fund via the San
Francisco General Foundation for funding this project.
We thank Robert Gallo and Nina Singh for helping review the final analyses and providing clinical feedback.

\section*{Ethics Statement}
This study was conducted as a quality improvement initiative at Zuckerberg San Francisco General Hospital and
Trauma Center. This study was approved by the University of California, San Francisco Institutional Review
Board (protocol 22-36613) as well as Zuckerberg San Francisco General Hospital. Individual informed consent was
waived, given the retrospective nature of the analysis and the use of deidentified electronic health record data. All
data were stored and analyzed within institutional systems in compliance with HIPAA regulations.

 \end{document}